\documentclass[12pt]{article}
\usepackage[margin=1in]{geometry}
\usepackage{amsmath,amssymb,amsthm}
\usepackage{natbib}
\usepackage[hypertexnames=false]{hyperref}
\usepackage{tikz}
\usepackage{enumitem}
\setlist{topsep=3pt, itemsep=2pt, parsep=0pt}
\usepackage{setspace}
\makeatletter
\def\thm@space@setup{\thm@preskip=7pt plus 2pt minus 3pt \thm@postskip=\thm@preskip}
\makeatother
\AtBeginDocument{%
\setlength{\abovedisplayskip}{4pt plus 2pt minus 2pt}%
\setlength{\belowdisplayskip}{4pt plus 2pt minus 2pt}%
\setlength{\abovedisplayshortskip}{2pt plus 1pt}%
\setlength{\belowdisplayshortskip}{3pt plus 1pt minus 2pt}}

\newtheorem{proposition}{Proposition}
\newtheorem{theorem}{Theorem}
\newtheorem{lemma}{Lemma}
\newtheorem{corollary}{Corollary}
\newtheorem{remark}{Remark}

\title{Search Contests}
\author{Emre Ozdenoren\thanks{London Business School, Regent's Park, London NW1 4SA, UK. Email: \href{mailto:eozdenoren@london.edu}{eozdenoren@london.edu} (corresponding author)} \and Murat Erkurt\thanks{Independent researcher.} }
\date{\today}

\begin{document}
\maketitle

\begin{abstract}
We study contests in which players search sequentially, drawing from a distribution at a cost per draw, and the highest accepted value wins the prize. We analyze unlimited and bounded search, with and without recall. Equilibria are distribution-free, depending only on the number of players, search depth, and the cost-prize ratio. With multiple prizes, equilibrium depends only on the average and the last-place prize. Entry has two effects, discouragement, which lowers the acceptance threshold, and selectivity, which raises it. With unlimited search, competition dissipates all rents, leaving only discouragement. With bounded search, rents survive, selectivity dominates at low cost, and the threshold is single-peaked in the number of players. In large fields discouragement prevails. A planner's optimal prize is invariant to location shifts and rises with tail thickness. Two players are optimal with unlimited search, while heavy tails favor larger fields when search is bounded.
\end{abstract}

\noindent\textbf{Keywords:} Contests, sequential search, rent dissipation, optimal stopping, tournaments.

\noindent\textbf{JEL Codes:} C72, D44, D83.

\section{Introduction}

There are many settings in which contenders compete for a prize through costly search: firms develop prototypes and a sponsor rewards the best design, or competing firms each hire for a key position and the firm with the best hire wins the business. Each attempt is costly and produces an outcome of uncertain quality, and a fixed reward goes to whoever ends up with the best result. Environments like these combine two classical frameworks: the theory of optimal search \citep{McCall1970, Stigler1961}, in which a single agent draws offers from a distribution and decides when to stop, and the theory of contests \citep{LazearRosen1981, Tullock1980}, in which agents choose effort levels and compete for a prize. Other environments are similar in spirit if not in every detail. When AI laboratories tune model architectures or drug companies screen compounds, beating the competition matters, but the prize is not fixed: what a firm earns depends on the quality of the draw itself, and more than one firm can profit. We fix the prize to set this dependence aside: a draw then pays only by beating the other players' draws, and the model isolates the effect of competition on search in its simplest form.

In our model, $N$ risk-neutral players independently draw values from a continuous distribution~$F$ at a cost $c>0$ per draw; we refer to the number of players $N$ as the field size. The player with the highest accepted value wins a fixed prize $W>0$; all others receive nothing. Each player's problem is a sequential search problem in which the value of stopping depends on the behavior of competitors.

The model has two structural parameters: the search depth $k$ (the maximum number of draws), and the recall convention (whether a player who continues may return to a previously rejected draw). With unlimited draws ($k=\infty$), competition drives continuation values to zero, so previously rejected draws are never recalled and the recall convention is irrelevant. With finitely many draws the convention matters, and we solve both cases. 

Of the two conventions, recall is the natural one when search produces something the searcher keeps, as with prototypes or submitted solutions, so the case with no recall deserves a word of motivation. No recall fits settings in which a draw passed over may be gone when the searcher comes back for it, and it is the premise of the classical secretary problem: an employer who passes on a candidate loses her, because the candidate is absorbed by the wider market. Our contest is $N$ such employers, each filling one position from its own applicant pool before a deadline, whose hires then compete: the firm with the best hire wins the prize. Whether recall is possible or not depends on market conditions: in a slack market a passed candidate can often be called back, in a tight one never. 

Consider first the unlimited-draw contest. In its unique symmetric equilibrium, each player accepts a draw if and only if it exceeds a cutoff~$\lambda$ determined by $1-F(\lambda) = Nc/W$. The threshold depends on~$F$, but the acceptance probability does not: it is pinned down by the number of players, the search cost, and the prize alone. The expected number of draws, the expected search cost, and all payoff-relevant quantities are therefore distribution-free, determined by the observables $N$, $c$, and $W$.\footnote{By contrast, a single searcher's acceptance probability depends on the distribution: with values uniform on $[0,1]$, any given draw is accepted with probability $\sqrt{2c}$; with unit-mean exponential values, with probability $c$; and fatter tails push the searcher into ever higher quantiles, because the option value of continued search is higher. Competition removes this dependence by driving the option value of search to zero.} The condition $1-F(\lambda) = Nc/W$ also shows the effect of entry: one more competitor raises the required acceptance probability. With unlimited draws, increased competition discourages players from searching by lowering the threshold.

The equilibrium condition is a zero-profit condition, and full rent dissipation follows: each player's expected payoff is exactly zero, and total search expenditure across all players equals the prize~$W$, for any~$F$ and whenever $Nc < W$. The dissipation follows the familiar logic of all-pay competition \citep{BayeKovenockdeVries1996, Siegel2009}, but takes a different form here: it obtains in a unique equilibrium in pure stationary threshold strategies, where the deterministic all-pay auction dissipates rents only in mixed strategies.

With finitely many draws, players cannot compete profits away, and rents survive. Rents change how the threshold responds to entry. The effect of one more competitor is the sum of a negative dilution term, from sharing the prize among more claimants, and a positive rent term, proportional to equilibrium rents, from the more aggressive gamble a player must take to win them. Under recall the equilibrium is a single stationary threshold, unique and distribution-free at every depth, and it is single-peaked in the number of players. At low search costs the rent term dominates: entry makes each player more selective, which corrects a comparative static in \citet{Taylor1995}. At high costs the dilution term dominates and entry lowers the threshold, and the same discouragement arrives in large fields at any fixed cost, because entry itself erodes the rents.

Without recall the strategy is a vector of round thresholds rather than a single number, and the analysis is harder. We show that a symmetric equilibrium exists at every depth, with thresholds declining across rounds as the forced final draw approaches. At two and three draws we solve the contest fully. At two draws the equilibrium is a single threshold pinned down by one equation, and it is single-peaked in the number of players, as under recall. At three draws the equilibrium has two thresholds. We show that the equilibrium is unique at every admissible cost and that both thresholds rise with entry at low cost. So the selectivity effect holds even without recall, for $k = 2$ and $3$. Two limit theorems then cover large fields at every depth. With free search, sufficiently large entry strictly raises every threshold, and every threshold approaches the top of the distribution at rate $\ln N/N$.\footnote{The result does not conflict with the recall case, where at any fixed positive cost entry eventually lowers the threshold in large fields: that comparative static concerns $c > 0$, while here $c = 0$. At zero cost the recall contest is degenerate: redrawing is weakly dominant and every player searches to full depth. Without recall the zero-cost contest remains interior, and the effect of entry is a real question.} With costly search, all thresholds approach zero as $Nc$ approaches the total prize $W$.

The structure survives two extensions. With rank-order prizes, we show that equilibrium depends only on the average prize and the consolation prize received by the last-place contestant. How the budget is divided among the higher ranks is irrelevant, and with no consolation prize any prize structure replicates winner-take-all. And when $M$ designers each field a team of $N$ searching workers and compete for a meta-prize, the designer contest is again distribution-free, with partial rent dissipation.

To study efficiency, we ask what prize a planner should post. Unlike the players, the planner cares about the actual quality of the winning output, so the value distribution, which drops out of equilibrium behavior, re-enters the design problem. Adding a constant to every value leaves the efficient prize $W^*$ unchanged, while thickening the upper tail at a fixed mean strictly raises it, without bound on the Pareto family.

A planner choosing both the prize and the field size faces a tradeoff between depth and breadth: fewer workers searching more selectively, or more workers searching less selectively. With unlimited draws, depth always dominates: $N^*=2$ regardless of the distribution. When search depth is bounded, heavy-tailed distributions break this ranking and larger fields become optimal. The pattern resonates with settings where time pressure and fat-tailed payoffs coexist: during the COVID-19 pandemic, Operation Warp Speed funded dozens of vaccine candidates in parallel, while breakthroughs in basic science, such as CRISPR gene editing, typically emerge from small groups with long time horizons. The model's predictions, linking optimal field size to tail thickness and search depth, are in principle testable.

\paragraph{Related literature.} Our paper connects two large literatures. On the search side, \citet{McCall1970} and \citet{Mortensen1986} study sequential search by a single agent, and \citet{Weitzman1979} solves the foundational problem of a planner who sequentially searches across heterogeneous alternatives to find the best one. On the contest side, \citet{LazearRosen1981}, \citet{Tullock1980}, and \citet{MoldovanuSela2001} study one-shot effort provision in competitions. \citet{MoldovanuSela2001} show that winner-take-all maximizes total effort with linear or concave costs; in our multiple-prize extension the prize structure is instead neutral: any rank-order allocation of a fixed budget with no consolation prize induces the same threshold, the same submitted values, and the same total expenditure as winner-take-all.

The closest paper is \citet{Taylor1995}, whose research tournament is our finite-depth contest with recall: firms search at a per-period cost over a finite horizon, and he characterizes the threshold equilibrium and the sponsor's design problem. His central design finding is ``aim short'': the sponsor-optimal quality threshold lies below the first-best, protecting against the risk that firms fall short. Our field-size result responds to the same shortfall risk on a different margin: more players means more independent shots at the upper tail. Within his model, our contribution is the transformation to quantile space, which makes the equilibrium distribution-free: the acceptance quantile depends only on $N$ and $c/W$. The formulation is what delivers the corrected comparative static in $N$ and the design results linking the efficient prize and the optimal field size to tail thickness rather than to the mean. Search without recall lies entirely outside his framework, as does the unlimited-draw benchmark, which he notes in a footnote but does not analyze.

The selectivity-versus-discouragement comparative static also corrects an error in \citet{Taylor1995}. The discussion following his Proposition~3 asserts that the equilibrium threshold is strictly decreasing in the number of contestants. The derivation, equation~A5 of his technical appendix, keeps the dilution term but omits the term generated by the dependence of the rival-outcome distribution on the number of contestants. At low cost the omitted term dominates and the correct derivative is positive. Section~\ref{sec:rents} derives the corrected derivative in quantile space and shows that the omitted term is the rent term: it is proportional to equilibrium rents, so selectivity is possible exactly when rents are positive.\footnote{\citet{BrookinsUsvitskiyTaylor2026} independently identify the same omission in a continuous-time version of the model and derive the corrected comparative static there. Section~\ref{sec:rents} discusses the connection.}

\citet{SeelStrack2013} analyze a contest in which players privately observe Brownian motions and choose when to stop; the equilibrium involves mixed stopping strategies, and in the martingale case the equilibrium distribution is independent of volatility. Our model differs in that draws are i.i.d., costs are explicit, and the equilibrium is in pure threshold strategies, with an acceptance probability independent of the entire distribution. \citet{Whitmeyer2017} studies a discrete-time model with finitely many rounds, i.i.d.\ uniform draws, no search cost, and no recall; his equilibrium depends on the uniform distribution assumption. Single-agent search with uncertain recall goes back to \citet{KarniSchwartz1977}, with deadlines added by \citet{AkinPlatt2014}. In secretary problems with competing employers \citep{ImmorlicaKleinbergMahdian2006}, several employers interview one common stream of applicants, observe only relative ranks, and race to secure the single best candidate, so the competition is over who gets each applicant; our employers draw cardinal values from independent pools and compete only after hiring, through the quality of the hire.

Our model is formally close to the all-pay auction literature. \citet{BayeKovenockdeVries1996} show that in symmetric complete-information all-pay auctions, total expenditure dissipates the entire rent. \citet{Siegel2009} generalizes this class: players buy scores directly, at player-specific costs, and the highest scores win. His main theorem determines every player's payoff without solving for equilibrium play: competition drives each winner down to the give-up point of his strongest excluded rival, so identical players, having no advantage to protect, earn zero. Our players cannot buy a score: a stopping rule delivers a random score and a random total cost. With unlimited draws, rent dissipation follows the same zero-profit logic. Tracing its implications through the search technology, however, yields a different kind of equilibrium: in the deterministic auction dissipation occurs only in mixed strategies, while the search contest has a unique equilibrium in pure stationary thresholds, in closed form, with the stopping margin pinning the acceptance probability at $Nc/W$, independent of the distribution. The correspondence is confined to unlimited draws: with bounded depth search cannot be scaled to the zero-profit point, rents survive, and they govern the comparative statics (Section~\ref{sec:rents}). \citet{Loury1979} studies R\&D races in which free entry leads to overinvestment; our planner's result ($N^*=2$) confirms that free entry would generate excessive parallel search.

Our model relates to the literature on strategic experimentation \citep{BoltonHarris1999, KellerRadyCripps2005}, where agents learn from each other's outcomes and competition creates a free-riding incentive. \citet{HoppeKatsenosOzdenoren2023} study a related preemption game and show that rivalry reduces experimentation. In our model there is no informational externality: draws are private and i.i.d., so competition reduces search through rent dissipation rather than strategic learning.

More broadly, the paper is related to the literature on innovation contests \citep{CheGale2003, HalacKartikLiu2017, FullertonMcAfee1999}. \citet{CheGale2003} and \citet{FullertonMcAfee1999} find that restricting entry to two firms is optimal, via incentive-dilution arguments among heterogeneous agents; our $N^*=2$ result arises from a purely statistical argument about order statistics. The value of parallel exploration dates to \citet{Nelson1961}, and \citet{DahanMendelson2001} tie the optimal number of parallel concept tests to the tail of the payoff distribution in a non-strategic setting; our field-size results embed the tradeoff in a competitive search equilibrium. \citet{TerwieschXu2008} decompose the effect of adding contestants into a negative incentive effect and a positive parallel-path effect; \citet{BoudreauLaceteraLakhani2011} provide empirical evidence from TopCoder programming contests. Our finite-draw analysis provides a theoretical counterpart: selectivity dominates discouragement when search costs are low and depth is limited.

The depth-vs-breadth tradeoff connects to a broader literature on parallel exploration. \citet{Nelson1961} and \citet{DahanMendelson2001} study parallel R\&D and concept testing; \citet{DahanMendelson2001} show that the optimal number of parallel tests depends on the tail shape of the payoff distribution. \citet{ChadeSmith2006} solve a related portfolio-search problem in a centralized framework. \citet{DrugovRyvkin2020} show that the IFR/DFR hazard-rate classification drives optimal prize allocation across ranks: our paper uses the same classification for a different design question, namely the efficient prize level and optimal field size.

\paragraph{Roadmap.} Section~\ref{sec:model} presents the model. Section~\ref{sec:equilibrium} derives the equilibrium of the unlimited-draw contest. Section~\ref{sec:finite_recall} analyzes finitely many draws under recall. Section~\ref{sec:no_recall} analyzes the finite-draw contest without recall. Section~\ref{sec:planner} analyzes the planner's problem. Section~\ref{sec:designers} extends the model to hierarchical competition. Section~\ref{sec:conclusion} concludes.

\section{Model}\label{sec:model}

There are $N\geq 2$ risk-neutral players; we refer to the number of players $N$ as the field size. Each player may take up to $k$ draws from a distribution $F$ with continuous, strictly positive density $f$ on the support $[0,K]$, where $K\leq\infty$, paying a cost $c>0$ per draw. After observing a draw $x$, the player chooses whether to stop (submit a value) or continue (pay $c$ to take the next draw, if any remain). 

We consider two conventions: no recall, where a player who continues discards the current draw and must submit whatever she ultimately stops on, and recall, where a player who stops at round $j$ submits the maximum of her first~$j$ draws. Under $k=\infty$ the player may continue indefinitely; under finite $k$ a player who continues past round $k-1$ must stop at round $k$. Players do not observe each other's draws or stopping decisions. After all players have stopped, the player with the highest submitted value wins the prize $W>0$. All others receive $0$. Ties, which occur with probability zero, are broken uniformly at random.

The two recall conventions yield identical equilibria under $k=\infty$ (Remark~\ref{rem:recall_inf}). Hence, we analyze three cases: the unlimited-draw contest with $k=\infty$ (the benchmark case), finitely many draws with recall, and finitely many draws without recall.

A threshold strategy under recall is a single number: the player stops as soon as her running maximum reaches it. If her draws run out first, she stops at the final round and submits the running maximum. A threshold strategy under no recall assigns a threshold to each of the first $k-1$ rounds: the player stops at a draw if and only if it reaches the current round's threshold. The threshold of the final round is zero by definition: a player who reaches it accepts her last draw. With unlimited draws the continuation problem is the same at every round, and a single stationary threshold suffices under either convention. We look for symmetric Nash equilibria in threshold strategies: profiles in which all players use the same thresholds.

We assume $Nc<W$.\footnote{When $Nc = W$, the equilibrium has $\lambda=0$ (or the infimum of the support): each player draws exactly once and accepts whatever she gets. When $Nc > W$, there is a symmetric mixed-entry equilibrium in which each player enters with probability $p^* \in (0,1)$: the positive payoff from being the sole entrant balances the negative payoff that arises when more than $W/c$ players enter. The assumption $Nc < W$ ensures full participation.} Proofs of all results are in the appendix; the proofs for the three-draw contest and the derivation of the designer first-order condition are in the Online Appendix.

\section{Infinite-Draw Contest with or without Recall}\label{sec:equilibrium}

\subsection{Equilibrium}

We look for a symmetric equilibrium in stationary threshold strategies. Each player faces an optimal stopping problem whose continuation value depends on opponents' behavior. In a symmetric equilibrium, competition pins down the common threshold through a zero-profit condition, which we derive below.

Suppose all players other than player~$i$ use threshold $\lambda^*$. Each opponent's final (accepted) value has the truncated distribution
\[
G(x) = \frac{F(x)-F(\lambda^*)}{1-F(\lambda^*)}, \qquad x\geq\lambda^*,
\]
with $G(x) = 0$ for $x < \lambda^*$.
If player $i$ stops at value $x\geq\lambda^*$, she wins if and only if her value exceeds all $N-1$ opponents' values. The probability of winning is $G(x)^{N-1}$.

Player $i$ faces a sequential search problem with ``wage'' $W\cdot G(x)^{N-1}$ for stopping at $x$. The Bellman equation for the continuation value $V$ is
\begin{equation}\label{eq:bellman}
V = -c + \int_0^K \max\bigl\{W\cdot G(x)^{N-1},\;V\bigr\}\,dF(x).
\end{equation}
The optimal strategy is a threshold: accept $x$ if and only if $W\cdot G(x)^{N-1}\geq V$. At the threshold $\lambda$, indifference requires $W\cdot G(\lambda)^{N-1}=V$.

In a symmetric equilibrium, $\lambda=\lambda^*$. Then $G(\lambda)=(F(\lambda)-F(\lambda))/(1-F(\lambda))=0$, so
\begin{equation}\label{eq:Vzero}
V = W\cdot 0^{N-1} = 0 \qquad (N\geq 2).
\end{equation}
The continuation value is zero: competition eliminates all rents from search.

Substituting $V=0$ into the Bellman equation~\eqref{eq:bellman} and using $\int_\lambda^K G(x)^{N-1}\,dF(x) = (1-F(\lambda))/N$ (by the substitution $u = G(x)$) yields
\begin{equation}\label{eq:equilibrium}
1-F(\lambda) = \frac{Nc}{W}.
\end{equation}

\begin{remark}\label{rem:recall_inf}
The equilibrium is the same whether or not players can recall previously rejected draws: all rejected draws lie below~$\lambda$ and are dominated by the stopping draw, so the option to recall is never exercised, and all equilibrium quantities are identical under both conventions.
\end{remark}

\subsection{Uniqueness}\label{sec:uniqueness}

We next show that there is a unique symmetric equilibrium. Asymmetric equilibria can exist for $N \geq 3$, but since players are ex ante symmetric, we focus on the symmetric case throughout.

\begin{proposition}\label{prop:equilibrium}
There exists a unique symmetric equilibrium. Each player uses threshold $\lambda = F^{-1}(1-Nc/W)$: accept a draw if and only if it exceeds $\lambda$.
\end{proposition}

Existence follows from the derivation preceding~\eqref{eq:equilibrium}; uniqueness follows because~\eqref{eq:equilibrium} has a unique solution in $\lambda$.

\begin{remark}\label{rmk:asymmetric}
For $N=2$, there are no asymmetric equilibria: the symmetric equilibrium is the unique Nash equilibrium.\footnote{Suppose, for contradiction, $p_1 > p_2$, where $p_i = 1-F(\lambda_i)$. Player~1 has the lower threshold; submitting at $\lambda_1$ never beats Player~2 (whose submission lies above $\lambda_2 > \lambda_1$), so $V_1 = 0$. Her Bellman equation reduces to $0 = -c + Wp_2/2$, giving $p_2 = 2c/W$. Player~2 earns $V_2 = W(p_1 - p_2)/p_1 > 0$, and her Bellman equation yields $Wp_2^2/(2p_1) = c = Wp_2/2$, which forces $p_1 = p_2$, a contradiction.} For $N\geq 3$, asymmetric equilibria exist. In such equilibria, passive players use a lower threshold and earn $V=0$, while active players use a higher threshold and earn strictly positive rents. The passive players are indifferent between accepting and rejecting draws below the active players' threshold, since both yield a payoff of zero.

These asymmetric equilibria are not robust: a passive player is indifferent over every threshold below the active players', and if active players occasionally slip and submit below it, her best response is to raise her threshold, undermining the equilibrium.
\end{remark}

\subsection{Distribution-Free Equilibrium Properties}

In our model, acceptance rates, search durations, and total expenditures depend only on the observables $N$, $c$, and $W$, and are therefore in principle testable without distributional assumptions.

In the unique equilibrium, the acceptance probability per draw is $1-F(\lambda) = Nc/W$, the expected number of draws per player is $W/(Nc)$, the expected search cost per player is $W/N$, total expected search cost is $W$, and each player's expected payoff is zero: every one of these quantities depends only on $(N, c, W)$ and is independent of~$F$.

The distribution-free property is a direct consequence of competition driving the continuation value to zero. In a non-competitive search problem, the reservation value equates the cost $c$ to the option value of continued search, $\int_\lambda^K (x-\lambda)\,dF(x)$, which depends on $F$. In the contest, competition drives $V$ to zero, replacing the option-value condition with the zero-profit condition~\eqref{eq:equilibrium}, which involves $F$ only through the single quantity $F(\lambda)$.

Full rent dissipation, total search expenditure equal to the prize, connects our model to the rent-dissipation results in the auction literature. In symmetric complete-information all-pay auctions \citep{BayeKovenockdeVries1996}, dissipation arises through mixed bidding. In our sequential-search contest it arises with pure threshold strategies, driven by the zero-profit condition at the threshold, which pins expected search cost per player to $W/N$.

\subsection{Comparative Statics}

The zero-profit condition yields clean comparative statics results. The most substantive is the discouragement effect: adding competitors lowers the acceptance threshold.

\begin{corollary}\label{cor:compstat}
In the unique equilibrium:
\begin{enumerate}
\item (Discouragement) The acceptance probability is increasing in $N$: each player searches less when there are more competitors. The expected number of draws per player, $W/(Nc)$, is decreasing in $N$.
\item The acceptance probability is increasing in $c$ and decreasing in $W$: higher costs reduce search, and higher prizes increase search.
\item Total search expenditure equals $W$ regardless of $N$ and $c$.
\end{enumerate}
\end{corollary}

The discouragement effect has a clean interpretation: with more competitors, each player's chance of winning falls, reducing the expected return to search. In response, players lower their thresholds and search less. But since there are more players, total expenditure remains at~$W$: the extensive and intensive margins exactly offset. However, as we show in Section~\ref{sec:finite_recall}, this conclusion is special to unlimited draws.

\subsection{Multiple Prizes}\label{sec:multiprize}

Suppose there are $N$ players and rank-order prizes $W_1\geq W_2\geq\cdots\geq W_N\geq 0$: the player with the $j$th-highest accepted value receives $W_j$. Let $\bar W = \frac{1}{N}\sum_{j=1}^N W_j$ denote the average prize. We assume $\bar W > W_N$ (not all prizes are equal; otherwise there is no incentive to search) and $c < \bar W - W_N$, so that the equilibrium threshold below is interior.

\begin{proposition}\label{prop:multiprize}
In the unique symmetric equilibrium with rank-order prizes:
\begin{enumerate}
\item The equilibrium threshold satisfies
\begin{equation}\label{eq:multiprize}
1-F(\lambda) = \frac{c}{\bar W - W_N}.
\end{equation}
The acceptance probability $1-F(\lambda)$ depends on the prizes and cost, but not on~$F$.
\item Each player's expected payoff is $V = W_N \geq 0$.
\item Total search expenditure is $\sum_{j=1}^N (W_j-W_N)$. The dissipation ratio is
\[
\frac{\text{Total cost}}{\sum_j W_j} = 1-\frac{W_N}{\bar W}.
\]
\end{enumerate}
\end{proposition}

\begin{corollary}[Prize-structure neutrality]\label{cor:neutrality}
The equilibrium threshold, the distribution of submitted values, each player's payoff, and total search expenditure depend on the prize vector only through the average prize $\bar W$ and the consolation prize $W_N$. In particular, when $W_N = 0$, every rank-order allocation $W_1 \geq \cdots \geq W_N = 0$ of a fixed budget induces the same equilibrium as winner-take-all: the division of the budget among ranks $1$ through $N-1$ is payoff-irrelevant.\footnote{Weak monotonicity of the prizes is essential: it makes the expected prize from stopping increasing in the submitted value, so that optimal stopping takes threshold form. With prizes non-monotone in rank, the equilibrium would no longer be in threshold strategies.}
\end{corollary}

A draw above the common threshold is equally likely to land at any rank, so the marginal return to one more draw is the average prize $\bar W$, while stopping at the threshold guarantees last place and $W_N$. No other feature of the prize vector matters for equilibrium outcomes. Players earn rents equal to $W_N$, and all remaining prize money is dissipated through search. This result contrasts with static contests, where the division of the budget across ranks is the key design question \citep{MoldovanuSela2001}.

\section{Finite-Draw Contest with Recall}\label{sec:finite_recall}

The setup in this section is the research-tournament model of \citet{Taylor1995} with no entry fee. Within this model, the finite-depth contest with recall, our contribution is twofold. First, we show that the equilibrium is distribution-free when written in quantile space. Second, we prove that the equilibrium threshold is single-peaked in the number of competitors, rising at low cost and falling at high cost, which corrects an error in his analysis. The other parts of the paper, the unlimited-draw contest, the contest without recall, and the design problems, have no counterpart in his framework.

In Section \ref{sec:equilibrium} with unlimited draws, the equilibrium threshold decreases in $N$ for all $c>0$: more competitors means less search. Entry in fact pushes the threshold in two directions at once. An additional competitor lowers each player's chance of winning, which reduces the return to search and pushes the threshold down: this is the discouragement effect. An additional competitor also raises the bar for winning: a player whose continuation value is positive responds by rejecting draws she would otherwise accept to defend her rent by searching harder. With unlimited draws, full dissipation pins the continuation value to zero, and only discouragement operates. With finitely many draws dissipation is partial: players cannot keep redrawing to compete away rents, the continuation value stays positive, and the second force appears. We call it the \emph{selectivity effect}. Selectivity dominates at low search costs and discouragement at high costs.

Tractability of the finite-$k$ analysis depends on the recall convention. Under recall the optimal strategy reduces to a single stationary threshold, and clean results are available for every $k \geq 2$. This section treats the recall case. Under no recall the strategy is a $(k-1)$-dimensional threshold vector and the analysis is harder: Section~\ref{sec:no_recall} proves existence at every depth, solves two and three draws completely, and characterizes large fields at every depth.

\subsection{Setup and Equilibrium Condition}

Each player draws up to $k$ times with recall: if she stops at round $t$, her submission is $\max(x_1, \ldots, x_t)$. The following lemma restricts attention to stationary threshold strategies.

\begin{lemma}\label{lem:stationary}
If all other players use a stationary threshold strategy, then each player's best response is also a stationary threshold strategy.
\end{lemma}

Stationarity follows because the one-step gain from continuing has a single cutoff in the running maximum: above the cutoff even one more draw is not worth its cost, below it one more draw is already worthwhile, and induction makes the cutoff independent of the number of rounds remaining.

Under the symmetric stationary strategy $\lambda$ with $a = F(\lambda)$, an opponent's submitted quantile has CDF
\begin{equation}\label{eq:HR}
H_k^R(u; a) = \begin{cases} u^k & \text{if } u < a, \\[2pt] a^k + (u - a)\,\dfrac{1 - a^k}{1 - a} & \text{if } u \geq a. \end{cases}
\end{equation}
For $u < a$, the formula $u^k$ is the probability that all $k$ draws lie at or below $u$, which forces no early stopping since every draw is also below $a$. For $u \geq a$, the formula combines two events: with probability $a^k$ the rival never stops early, so her submission (the max of $k$ draws all below $a$) is automatically at most $u$; with probability $(u-a)(1-a^k)/(1-a)$ she stops early at a draw in $[a, u]$. Conditional on early stopping, that first-exceeding draw is uniform on $[a, 1]$. The expression is independent of $F$.

The indifference condition at round-1 threshold $a$ is
\begin{equation}\label{eq:gkR}
\frac{c}{W} \;=\; g_k^R(N, a) \;\equiv\; (1 - a)\left[\frac{1 - a^{kN}}{N(1 - a^k)} - a^{k(N-1)}\right],
\end{equation}
which does not depend on $F$.\footnote{At $q_1 = a$, stopping gives win probability $H_k^R(a;a)^{N-1} = a^{k(N-1)}$ and payoff $W a^{k(N-1)}$. Continuing costs $c$, after which (by Lemma~\ref{lem:stationary}) the player stops at round~2 with submission $\max(a, q_2)$ since the running maximum already meets the threshold. Expected continuation payoff: $-c + W[a \cdot a^{k(N-1)} + \int_a^1 H_k^R(u;a)^{N-1} du]$. The integral equals $(1-a)(1-a^{kN})/(N(1-a^k))$ by the substitution $v = a^k + (u-a)(1-a^k)/(1-a)$. Setting stop = continue and solving for $c/W$ gives~\eqref{eq:gkR}.} As in the unlimited-draw case, the equilibrium is distribution-free: the acceptance quantile $a^*(N, k, c/W)$ is a function of $(N, k, c/W)$ alone.

\begin{proposition}\label{prop:finite_k_recall_eq}
For every $c/W \in (0, 1/N)$, the $k$-draw contest with recall has a unique equilibrium. It is symmetric: each player uses the stationary threshold $\lambda^* = F^{-1}(a^*)$, where $a^*(N, k, c/W) \in (0, 1)$ is the unique solution of $c/W = g_k^R(N, a)$.
\end{proposition}

At the boundary $c = 0$, recall makes redrawing weakly dominant, so $a^*(N, k, 0) = 1$ (every player uses her full search depth). As $Nc \uparrow W$, $a^* \downarrow 0$. We restrict attention to $Nc < W$ throughout.

\subsection{Selectivity versus Discouragement}\label{sec:rents}

As a preliminary step we derive the players' equilibrium payoff in closed form. Denote a player's expected search cost as $K(a) = c(1-a^k)/(1-a)$.

\begin{lemma}\label{prop:surplus}
In the $k$-draw contest with recall, each player's equilibrium payoff is
\begin{equation}\label{eq:surplus}
V \;=\; \frac{W}{N} - K(a^*) \;=\; W A^{N-1}\Bigl(1 - \frac{N-1}{N}\,A\Bigr), \qquad A = (a^*)^k,
\end{equation}
where the second equality substitutes the equilibrium condition~\eqref{eq:gkR}. The payoff is strictly positive at every interior equilibrium.
\end{lemma}

To interpret \eqref{eq:surplus} rewrite it as $W\bigl[A^{N-1}(1 - A) + A^N/N\bigr]$. The first term is the probability that the player draws above the common threshold while no rival ever does, so that she wins unchallenged. The second is the probability that nobody draws above the threshold, in which case the best of the below-threshold maxima wins, each player's with probability $1/N$. The rent is exactly the prize collected in these unchallenged events. In every other event some rival crosses the threshold, and the prize the player expects to win in those contested races is exactly offset by her expected search cost. As the search depth grows the unchallenged events become negligible and rents vanish: the bracket in~\eqref{eq:gkR} is at most one, so $c/W \leq 1 - a^*$ and $A \leq (1 - c/W)^k \to 0$, and the full-dissipation result of Section~\ref{sec:equilibrium} emerges as the limit. The dissipation ratio is $NK(a^*)/W = 1 - NA^{N-1} + (N-1)A^{N}$.

Identity~\eqref{eq:surplus} determines how the threshold responds to competition. Treating $N$ as a continuous variable and differentiating it implicitly yields the following result.

\begin{proposition}\label{prop:corrected}
At any interior equilibrium, with $N$ treated as a continuous variable,
\begin{equation}\label{eq:dadN}
\frac{da^*}{dN} \;=\; \frac{V \ln(1/A) \;-\; (W/N^2)\,(1 - A^N)}{K'(a^*) \;+\; W k (N-1)\, (a^*)^{k-1} A^{N-2} (1 - A)}.
\end{equation}

\end{proposition}

Since the denominator is strictly positive: the threshold rises with entry if and only if $V \ln(1/A) > (W/N^2)(1 - A^N)$. Two forces appear in the numerator of~\eqref{eq:dadN}. The negative term $(W/N^2)(1 - A^N)$ is the \emph{dilution term}: an additional competitor lowers each player's expected share of the prize. The positive term $V \ln(1/A)$ is the \emph{rent term}, proportional to the equilibrium rent: an additional competitor makes stopping at the boundary less safe, and the value of responding with further search scales with the rent at stake. Selectivity therefore requires rents. With unlimited draws rents are zero, only the dilution term remains, and the threshold falls in $N$ at every cost level, which is the discouragement result of Corollary~\ref{cor:compstat}. With finitely many draws rents are positive by Lemma~\ref{prop:surplus}, and at low cost ($A$ near one) the rent term dominates.

Proposition~\ref{prop:corrected} locates the error in \citet{Taylor1995} exactly. Equation~A5 of his technical appendix computes the group-size derivative with the dilution term only; the rent term is omitted, and the omission drives his claim that the threshold is strictly decreasing in the number of contestants. \citet{BrookinsUsvitskiyTaylor2026} independently identify the same omission and derive the corrected two-term derivative in a continuous-time Poisson version of the model, in which the rent term appears as a reduced-safety effect. The two corrections agree: their model is the continuum limit of ours.\footnote{Set $a = 1 - L/k$ in~\eqref{eq:gkR} and let $k \to \infty$: $k\, g_k^R(N, a) \to L\bigl[(1 - e^{-NL})/(N(1 - e^{-L})) - e^{-(N-1)L}\bigr]$, which is the equilibrium function of \citet{BrookinsUsvitskiyTaylor2026} with $kc/W$ in the role of their cost-to-prize ratio $cT/P$.}

The sign criterion also determines the global shape of the response to entry.

\begin{theorem}\label{thm:singlepeak}
Fix $k \geq 2$ and $c/W \in (0, 1/2)$, and treat $N$ as a continuous variable on the admissible range $[2, W/c)$. There exists $N^\dagger \in [2, W/c)$ such that the equilibrium quantile $a^*(N, k, c/W)$ is strictly increasing in $N$ on $[2, N^\dagger)$ and strictly decreasing on $(N^\dagger, W/c)$; either piece may be empty. In particular, the integer sequence $a^*(2), a^*(3), \ldots$ rises and then falls, with no second rise.
\end{theorem}

The proof rests on one observation: at any point where the numerator of~\eqref{eq:dadN} vanishes, its derivative along the equilibrium path is strictly negative, by the classical inequality $\ln(1/x) > 2(1-x)/(1+x)$ with $x = A^N$. Every critical point is therefore a strict downward crossing, and there can be at most one. \citet{BrookinsUsvitskiyTaylor2026} in a continuous-time model establish the same strict shape, strictly increasing up to a unique peak and strictly decreasing beyond it, on every connected interval of group sizes on which an interior equilibrium exists. Theorem~\ref{thm:singlepeak} establishes the shape at every finite depth $k$, of which their model is the continuum limit, and over the full admissible range $[2, W/c)$, where the equilibrium is interior throughout. In their model the peak group size falls as the horizon lengthens and converges to one, so that in long contests entry discourages search throughout the relevant range, in line with the unlimited-draw benchmark of Section~\ref{sec:equilibrium}.\footnote{They also decompose the sponsor's gain from a larger field into a direct participation effect and what they call the encouragement premium, and use the decomposition to study upfront investment and entry subsidies.}

The following proposition gives the two regimes, selectivity at low cost and discouragement at high cost, with explicit ranges.

\begin{proposition}\label{prop:finite_k_recall}
\begin{enumerate}
\item[(i)] For each finite $k \geq 2$ and each integer $\bar N \geq 2$, there exists $\bar c(\bar N, k) > 0$, independent of $F$, such that for all $c/W \in (0, \bar c(\bar N, k))$,
\[
a^*(2, k, c/W) \;<\; a^*(3, k, c/W) \;<\; \cdots \;<\; a^*(\bar N + 1, k, c/W).
\]
\item[(ii)] For every $c/W \in (0, 1/2)$, the sequence $\{a^*(N, k, c/W)\}$ is strictly decreasing on some terminal range of admissible $N$.
\end{enumerate}
\end{proposition}

Here $\bar N$ is an arbitrary upper bound on the length of the increasing chain: given any $\bar N$, the cost can be chosen small enough that selectivity holds at every step up to $\bar N + 1$. Part (ii) says that the rise always ends: at any fixed cost, entry eventually lowers the threshold, which falls toward zero as aggregate spending $Nc$ approaches the prize.

\section{Finite-Draw Contest without Recall}\label{sec:no_recall}

Under the no-recall convention a player who continues discards her current draw, so with $k$ draws the strategy is no longer a single number: a player chooses a threshold for every round, and the equilibrium is a $(k-1)$-dimensional vector. This section analyzes the model at every depth. At any $k$ we establish existence of a symmetric equilibrium and show that within any equilibrium the thresholds decline across rounds. At $k = 2$ and $k = 3$ we obtain a complete theory: closed form, uniqueness at every admissible cost, and selectivity of every threshold at low cost. Two limit theorems then settle the large-field behavior at every depth: with free search, sufficiently large entry strictly raises every threshold, while with costly search entry eventually destroys selectivity: every threshold tends to zero as $Nc$ approaches $W$.

\subsection{Setup, Existence, and Declining Thresholds}\label{sec:nr_setup}

We work in quantile space, $u = F(x)$ which makes all the results distribution-free in quantiles. A symmetric profile is a vector of round thresholds $a_1 \geq a_2 \geq \cdots \geq a_{k-1}$: round $r$ accepts quantiles $u \geq a_r$, and round $k$ must accept. Against opponents who use the profile, the probability that an opponent reaches round $r$ is $P_r = a_1 a_2 \cdots a_{r-1}$ (with $P_1 = 1$), and her submitted quantile has CDF
\begin{equation}\label{eq:nrH}
H(u) \;=\; \sum_{r=1}^{k-1} P_r\,\max\{u - a_r,\, 0\} \;+\; P_k\,u ,
\end{equation}
continuous, piecewise linear, and strictly increasing, with $H(0) = 0$ and $H(1) = 1$. Write $\Phi = H^{N-1}$ for the probability that a submission beats all $N-1$ opponents. Let $w_m$ denote a player's continuation value, in units of the prize, at the start of a round with $m$ draws remaining: she pays the cost and observes a draw, and then either submits it or continues with $m-1$ draws left. With one draw remaining she must submit. The values satisfy
\begin{equation}\label{eq:nrvalues}
w_1 = -\frac{c}{W} + \int_0^1 \Phi(u)\,du,
\qquad
w_m = -\frac{c}{W} + \int_0^1 \max\{\Phi(u),\, w_{m-1}\}\,du, \quad m \geq 2 .
\end{equation}
The recursion optimizes over all stopping rules, not only threshold rules: draws are i.i.d.\ and a rejected draw is payoff-irrelevant under no recall, so the round and the current draw form a sufficient state, and the optimal accept set in each round is an upper interval because $\Phi$ is nondecreasing. The acceptance threshold with $m \geq 2$ draws remaining is $\tau_m = \inf\{u : \Phi(u) \geq w_{m-1}\}$, equal to zero when $w_{m-1} \leq 0$, and a symmetric equilibrium is a profile that reproduces itself: $a_r = \tau_{k-r+1}(a_1, \ldots, a_{k-1})$ for every $r$, each threshold a best response to the submission distribution $H$ generated by the entire profile.

The first result shows that a symmetric equilibrium exists at every depth, every field size, and every admissible cost.

\begin{proposition}\label{prop:nr_exist}
For every $k \geq 2$, $N \geq 2$, and $c/W \in [0, 1/N)$, the $k$-draw contest without recall has a symmetric equilibrium.
\end{proposition}

The second result is the deadline effect: within an equilibrium, players grow less selective as the forced final draw approaches, because the option value of the remaining draws shrinks.

\begin{proposition}\label{prop:nr_decline}
Fix any opponent profile and define the values by~\eqref{eq:nrvalues}. Then $w_m \geq w_{m-1}$ for every $m \geq 2$, so any best response, and in particular any symmetric equilibrium, satisfies $a_1 \geq a_2 \geq \cdots \geq a_{k-1}$. If $w_1 > 0$, all inequalities are strict.
\end{proposition}

The analysis that follows works with the round conditions as equalities. Call an equilibrium \emph{interior} if $w_1 > 0$. By Proposition~\ref{prop:nr_decline} the values then all lie strictly between zero and one and $\Phi(a_r) = w_{k-r}$ for every $r$. If instead $w_1 \leq 0$, every threshold is zero; Section~\ref{sec:nr_limits} shows this cannot happen under our assumption $Nc < W$, so every symmetric equilibrium is interior.

\subsection{Two Draws}\label{sec:nr_k2}

\paragraph{Equilibrium.} With two draws the profile is a single threshold $a = a_1$, and the machinery of Section~\ref{sec:nr_setup} specializes in closed form. From~\eqref{eq:nrH} we obtain
\[
H(u) \;=\;
\begin{cases}
au & u \in [0, a],\\[2pt]
(1+a)u - a & u \in [a, 1].
\end{cases}
\]
Below the threshold, a submission lands at quantile $u$ only when round 1 is rejected and the forced second draw falls below $u$. Above it, accepted first draws contribute as well. Since $H(a) = a^2$, the round condition $\Phi(a) = w_1$ reads $a^{2(N-1)} = -c/W + \int_0^1 H(u)^{N-1}\,du$, and evaluating the integral over the two linear pieces of $H$,
\begin{equation}\label{eq:finite_threshold}
\frac{c}{W} \;=\; g(N,a), \qquad g(N,a) \;\equiv\; \frac{1+a^{2N-1}}{N(1+a)} - a^{2N-2}.
\end{equation}

For two draws equilibrium exists by Proposition~\ref{prop:nr_exist} and is unique as shown by the following proposition:

\begin{proposition}\label{prop:no_recall_eq}
For every $c/W \in [0, 1/N)$, the symmetric equilibrium of the two-draw contest without recall is unique: each player uses the round-1 threshold $\lambda^* = F^{-1}(a^*)$, where $a^*(N, c/W) \in (0, 1)$ is the unique solution of $c/W = g(N, a)$.
\end{proposition}

\paragraph{Selectivity versus discouragement.} The two opposing effects of Section~\ref{sec:finite_recall} also shape the threshold in the no-recall case. The substantive difference lies in the opportunity cost of continuing. Under no recall, rejecting round~1 means \emph{discarding} $q_1$ and committing to whatever $q_2$ delivers, so continuing has a built-in opportunity cost even when $c = 0$. Under recall, by contrast, continuing preserves $q_1$, so at $c = 0$ there is no cost to redrawing and the threshold degenerates ($a^* = 1$). Selectivity therefore bites at every $c \geq 0$ without recall, and the comparative statics below can be proved directly at $c = 0$ and extended by continuity. The next theorem is the counterpart of Theorem~\ref{thm:singlepeak} without recall.

\begin{theorem}\label{thm:nr_singlepeak}
Fix $c/W \in (0, 1/2)$, and treat $N$ as a continuous variable on the admissible range $[2, W/c)$. There exists $N^\dagger \in [2, W/c)$ such that $a^*(N, c/W)$ is strictly increasing in $N$ on $[2, N^\dagger)$ and strictly decreasing on $(N^\dagger, W/c)$; either piece may be empty. In particular, the integer sequence $a^*(2), a^*(3), \ldots$ rises and then falls, with no second rise.
\end{theorem}

The proof follows the crossing argument of Theorem~\ref{thm:singlepeak}, built on the two-draw counterpart of~\eqref{eq:surplus}:
\[
V \;=\; \frac{W}{N} - (1+a^*)\,c \;=\; \frac{W}{N}\,(a^*)^{2N-2}\bigl(N+(N-1)a^*\bigr).
\]
The new step is where this identity is used. Under recall, every critical point of the equilibrium equation crosses downward. Without recall some critical points do not, but at each of them the right side of the identity is at least $W/N$, so the left side would require zero or negative search spending leading to a contradiction with equilibrium. The proof shows that the remaining critical points all cross downward.

The next proposition gives the two regimes, selectivity at low cost and discouragement at high cost, as under recall.

\begin{proposition}\label{prop:finite_c}
\begin{enumerate}
\item[(i)] For each integer $\bar N \geq 2$, there exists $\bar c(\bar N) > 0$, independent of $F$, such that for all $c/W \in [0, \bar c(\bar N))$,
\[
a^*(2, c/W) \;<\; a^*(3, c/W) \;<\; \cdots \;<\; a^*(\bar N + 1, c/W).
\]
\item[(ii)] For every $c/W \in (0, 1/2)$, the sequence $\{a^*(N,c/W)\}$ is strictly decreasing on some terminal range of admissible $N$.
\end{enumerate}
\end{proposition}

Proof of part (i) uses a moment comparison: at the zero-cost equilibrium, stopping and redrawing produce the same $(N-1)$-th moment of the submission, and an extra competitor raises the exponent to $N$, which favors the more spread-out redraw by Lyapunov's inequality, so the threshold must rise to restore indifference. Part (ii) combines the theorem with the ceiling $a^*(N, c/W) \leq (W - Nc)/(Nc)$ shown in the appendix. 

\subsection{Three Draws}\label{sec:nr_k3}

With three draws the profile is a pair of thresholds $a_1 \geq a_2$, and the equilibrium is a two-dimensional fixed point. Compare the value of entering a round with two remaining draws to the value of entering with one: any draw above $a_2$ is accepted in both situations, so the extra draw pays only on draws below $a_2$. Below $a_2$ the opponents' submission CDF $H$ in~\eqref{eq:nrH} is a single linear segment, so the comparison evaluates in closed form, and through the round conditions it becomes a link between $a_1$ and $a_2$ involving neither the cost nor the distribution. One equation in one unknown remains, and its slope can be signed at every admissible cost. What is special about three draws is that every step of this analysis evaluates $H$ on a known segment: with a single ordered pair of thresholds, no comparison crosses a kink of $H$ from an undetermined side. Beyond three draws the arguments must compare profiles whose thresholds interleave with the kinks of $H$ in an order that is not pinned down in advance, the closed forms differ case by case, and the exact fixed-$N$ analysis remains open. The obstacle is technical rather than economic: no new force enters at four draws. The following theorem describes the equilibrium with three draws.

\begin{theorem}\label{thm:nr_k3}
For every $N \geq 2$ and $c/W \in [0, 1/N)$, the three-draw contest without recall has a unique symmetric equilibrium, and the acceptance quantiles depend only on $N$ and $c/W$. At $c = 0$, both thresholds strictly increase when a player joins; for each $N$ there is $\bar c_N > 0$, independent of $F$, such that the same holds for all $c/W \in [0, \bar c_N)$.
\end{theorem}

The proof is in the Online Appendix. Uniqueness follows from the reduction proved there (Proposition~OA.1): a single equation determines the round-2 threshold, and its slope is strictly negative at every admissible cost. Selectivity of the round-2 threshold follows from the same moment comparison as at two draws. Selectivity of the round-1 threshold is the delicate step: the equilibrium relation between the two thresholds shifts with entry toward a lower round-1 threshold, so the proof shows that the rise of the round-2 threshold more than offsets the shift.

\subsection{Large Fields}\label{sec:nr_limits}

We next characterize the equilibrium behavior for large fields sizes. With free search, entry makes every threshold rise, and all thresholds approach the top of the distribution at rate $\ln N/N$. With costly search, entry eventually destroys selectivity: every threshold is squeezed to zero as the aggregate search burden $Nc$ approaches the prize.

\begin{theorem}\label{thm:nr_largeN}
Fix $k \geq 2$ and $c = 0$.
\begin{enumerate}
\item[(i)] In any symmetric equilibrium, every acceptance quantile satisfies
\[
\bigl(k(N-1)+1\bigr)^{-\frac{1}{2(N-1)}} \;\leq\; a_r \;\leq\; 1 - \frac{1 - N^{-1/(N-1)}}{k} ,
\]
all thresholds converge to the top of the distribution at rate $\ln N / N$.
\item[(ii)] If $N \geq k^3$ and $M \geq k$, then every threshold of every symmetric equilibrium of the $MN$-player contest strictly exceeds every threshold of every symmetric equilibrium of the $N$-player contest. In particular $a_r$ rises for every round $r$: with free search, sufficiently large entry strictly raises every threshold, at every depth.
\end{enumerate}
\end{theorem}

The theorem is subtle. Even at zero cost, search without recall carries an implicit price: rejecting a draw means giving it up for good, and a player who holds out may end with something worse. Entry, however, erodes the value of settling faster than the value of continuing: in a large field a mediocre draw wins with negligible probability, so giving it up costs little, while continuing preserves a chance at the top. Part (i) makes the escalation quantitative: every threshold sits within distance of order $\ln N/N$ of the top of the distribution, so players in every round hold out for near-top draws, and the acceptance probability of any given draw vanishes as the field grows.

\begin{theorem}\label{thm:nr_collapse}
Fix $k \geq 2$ and $c > 0$.
\begin{enumerate}
\item[(i)] Every symmetric equilibrium is interior: every threshold is strictly positive.
\item[(ii)] In any symmetric equilibrium the top quantile satisfies
\[
\frac{a_1}{(1+a_1)^2} \;<\; 1 - \frac{Nc}{W},
\]
so every threshold tends to zero as $Nc \uparrow W$: with costly search, sufficiently large entry destroys selectivity, at every depth.
\end{enumerate}
\end{theorem}

The ceiling in part (ii) is the finite-depth face of the rent-dissipation logic of Section~\ref{sec:equilibrium}: each entrant burns at least $c$ against an aggregate prize $W$, and as $Nc$ approaches $W$ nothing remains to finance selectivity. The two theorems together give the large-field phase picture at every depth: selectivity survives entry exactly when search is free, because free search preserves the rents that finance it, and dies under any positive cost, with every threshold vanishing at $Nc = W$, the existence boundary of the unlimited-draw contest. The thresholds are single-peaked in $N$ under recall (Theorem~\ref{thm:singlepeak}) and at two draws without recall (Theorem~\ref{thm:nr_singlepeak}); whether the same shape holds at depths beyond two without recall is open.

\section{The Planner's Problem}\label{sec:planner}

We now turn to the planner's problem. The planner employs $N$ workers who search as in Section~\ref{sec:equilibrium} and offers a prize $W$ for the highest value. She faces two design choices. We study them sequentially: first the prize $W$ for a fixed field size $N$ (Section~\ref{sec:opt_prize}), which shapes how selectively each worker searches, and then the field size $N$ with the prize set optimally at $W^*(N)$ (Section~\ref{sec:opt_N}), which trades depth against breadth.

\subsection{The planner's objective}\label{sec:planner_obj}

Write $\mathbb{E}[\max] \equiv \mathbb{E}[\max(X_1,\ldots,X_N)]$ for the expected value of the best output. The planner values $\mathbb{E}[\max] - W$ (output quality minus the prize) and may also place weight $\alpha\in[0,1]$ on workers' surplus $W - \text{total search cost}$:
\begin{equation}\label{eq:SW}
\mathrm{SW} = \mathbb{E}[\max] - W + \alpha \bigl(W - \text{total search cost}\bigr) = \mathbb{E}[\max] - (1-\alpha)\,W - \alpha\cdot\text{total search cost}.
\end{equation}
When $\alpha = 0$, only output minus prize matters. When $\alpha = 1$, the prize is a pure transfer and only real resource costs matter.\footnote{\citet{Taylor1995} studies a related design problem in the finite-draw contest with recall, corresponding to our $\alpha = 1$ case: his sponsor extracts all worker surplus through entry fees, so her objective coincides with social welfare at $\alpha = 1$.}

In the unlimited-draw contest, $\alpha$ is irrelevant. Full rent dissipation makes total search cost equal to $W$, so the workers' surplus term $W - \text{total search cost}$ in~\eqref{eq:SW} is zero and the $\alpha$-weighted component drops out entirely: with unlimited draws, $\mathrm{SW} = \mathbb{E}[\max] - W$ for every $\alpha \in [0, 1]$.

In the finite-draw contest, dissipation is partial and the planner's objective depends on $\alpha$. Our finite-draw results in Sections~\ref{sec:opt_prize} and~\ref{sec:opt_N} hold for every $\alpha \in [0, 1]$.

Throughout the design analysis we assume that the expected winning output is finite and that the planner's problem attains its maximum at an interior threshold.

\subsection{Optimal prize}\label{sec:opt_prize}

For a given field size~$N$, the planner chooses the prize~$W$ to maximize social welfare. A higher prize raises the equilibrium threshold: workers become more selective, producing better expected output, but at greater cost. Unlike the competitive equilibrium, the efficient prize $W^*(F)$ is distribution-dependent: the planner must know~$F$ to set the right prize.

What is the right prize sensitive to? We restrict attention to the unlimited-draw contest, the finite-draw contest with recall, and the two-draw contest without recall, and establish two facts. First, the mean of $F$ plays no role: a pure location shift leaves $W^*$ unchanged. Second, the upper tail does: thickening the upper tail at fixed mean raises $W^*$, thinning it lowers $W^*$.

\begin{proposition}[Location invariance]\label{prop:location_inv}
Fix $N \geq 2$, $\alpha \in [0, 1]$, $c > 0$, and a base distribution $F_0$. In each of the three contests, if $F_\theta^{-1}(q) = F_0^{-1}(q) + \theta$ for $\theta \in \mathbb{R}$, then $W^*(F_\theta) = W^*(F_0)$.
\end{proposition}

Proposition~\ref{prop:location_inv} is the planner's analogue of the distribution-free equilibrium: workers' incentives depend on quantiles, not values, so a location shift leaves the quantile structure, and thus the marginal return to raising $W$, untouched.\footnote{In fact, the invariance Proposition~\ref{prop:location_inv} understates what holds under unlimited draws and recall. Under unlimited or finite draws with recall, any change to $F$ below the threshold is invisible, not just location shifts. Only the no-recall contest, where below-threshold values enter submissions, requires full location invariance.} 

To isolate the role of the upper tail at fixed mean, consider the \emph{mean-preserving tail-cap} family parameterized by a single cap level $M$:
\begin{equation}\label{eq:tail_family}
F_M^{-1}(q) \;=\; \min\bigl(F_0^{-1}(q),\, M\bigr) + \theta(M), \qquad M \in \bigl(\mu_0,\, F_0^{-1}(1^-)\bigr),
\end{equation}
where $\mu_0 = \int_0^1 F_0^{-1}(q)\,dq$ is the mean of $F_0$ and $\theta(M) = \mu_0 - \int_0^1 \min\bigl(F_0^{-1}(q), M\bigr)dq$ is the unique shift that preserves the mean at $\mu_0$ for every $M$. Lowering $M$ collapses upper-tail mass above $M$ to a point and lifts the rest of the quantile function uniformly by $\theta(M)$ to compensate. As $M \uparrow F_0^{-1}(1^-)$, $F_M \to F_0$; as $M \downarrow \mu_0$, $F_M$ degenerates to a point mass at $\mu_0$. Figure~\ref{fig:tail_family} illustrates.\footnote{The cap creates an atom of mass $1 - q_M$ at the value $M + \theta(M)$, which nominally falls outside the continuous-density assumption of Section~\ref{sec:model}. Uniform tie-breaking restores full validity of the quantile-space integrals without modification: for any $N$ players, the expected winning probability at the atom equals $[1 - (1-p)^N]/(Np)$ with $p = 1 - q_M$, which is exactly $(1/p)\int_{1-p}^1 u^{N-1}\,du$, the continuous-quantile value.}

\begin{figure}[h]
\centering
\begin{tikzpicture}[scale=0.85]
  \begin{scope}
    \clip (-0.1,-0.1) rectangle (7.3, 5.3);
    \draw[thick,blue,domain=0:0.985,samples=150,smooth]
      plot (\x*7, {(1-\x)^(-0.4)});
    \draw[thick,red,dashed,domain=0:0.637,samples=100,smooth]
      plot (\x*7, {(1-\x)^(-0.4) + 0.362});
    \draw[thick,red,dashed] (4.459,1.862) -- (7,1.862);
  \end{scope}
  \draw[->] (0,0) -- (7.4,0) node[right] {$q$};
  \draw[->] (0,0) -- (0,5.4) node[above] {$F^{-1}(q)$};
  \draw (7,0.08) -- (7,-0.08) node[below] {$1$};
  \draw[dotted,gray] (4.459,0) -- (4.459,1.862);
  \node[below] at (4.459,0) {$q_M$};
  \draw[dotted,gray] (0,1.5) -- (4.459,1.5);
  \node[left] at (0,1.5) {\small $M$};
  \draw[dotted,gray] (0,1.862) -- (4.459,1.862);
  \node[left] at (-0.04,1.95) {\small $M{+}\theta(M)$};
  \node[blue,right] at (7.05,4.95) {\small $F_0^{-1}$};
  \node[red,right]  at (7.05,1.862) {\small $F_M^{-1}$};
\end{tikzpicture}
\caption{\small The mean-preserving tail-cap family. The base $F_0^{-1}$ (solid) is capped at $M$ above quantile $q_M = F_0(M)$, and the entire quantile function is lifted by $\theta(M)$ to restore the original mean. The result $F_M^{-1}$ (dashed) lies above $F_0^{-1}$ in the body (lift) and flat below $F_0^{-1}$ in the upper tail (cap); the two changes exactly offset in expectation. As $M$ decreases, the cap bites more, the compensating lift grows, and the upper tail thins.}
\label{fig:tail_family}
\end{figure}

\begin{proposition}[Tail monotonicity]\label{prop:tail_mono}
Fix $N \geq 2$, $\alpha \in [0, 1]$, $c > 0$, and a base distribution $F_0$. In each of the three contests, $W^*(F_M)$ is strictly increasing in $M$ on $\bigl(F_0^{-1}(\bar q),\, F_0^{-1}(1^-)\bigr)$, where $\bar q = q^*(F_0)$, the acceptance quantile induced by the largest optimal prize at $F_0$, for the unlimited-draw and recall contests, and $\bar q = \bar a(N)$, the zero-cost equilibrium quantile, for the two-draw contest without recall.
\end{proposition}

Together, Propositions~\ref{prop:location_inv} and~\ref{prop:tail_mono} establish that the upper tail, not the location, drives $W^*$.\footnote{The cap family is bounded: as $M \uparrow F_0^{-1}(1^-)$, $W^*(F_M) \uparrow W^*(F_0)$, a finite ceiling. To exhibit \emph{unbounded} $W^*$, one varies the base itself. On the Pareto base $F_0^{(t)}(x) = 1 - x^{-1/t}$, as $t \to 1^-$ (the infinite-mean boundary), $W^*(F_0^{(t)}) \to \infty$ in each of the three contests.} 

\subsection{Optimal field size}\label{sec:opt_N}

We now add the field size $N$ to the planner's choices, with the prize set optimally at $W^*(N)$ for each $N$. More workers provide independent draws, but each additional worker forces a lower equilibrium threshold, trading depth (selectivity per worker) against breadth (number of workers). The balance between the two depends on whether each worker's search depth is bounded. For the unlimited-draw contest the result is distribution-free; for the finite-draw cases, we illustrate the break from that benchmark using the Pareto family as a canonical heavy-tail parameterization.

\begin{proposition}\label{prop:optimal_N}
Fix $\alpha \in [0, 1]$ and $c > 0$. Let $N^*$ denote the planner's optimal field size at $W = W^*(N)$.
\begin{enumerate}
\item[(i)] \emph{Unlimited draws.} For any continuous distribution~$F$, the planner's welfare $\mathrm{SW}^*(N)$ is strictly decreasing in~$N$. Hence $N^* = 2$, distribution-free.
\item[(ii)] \emph{Finite draws.} For each $N \geq 2$, there exists a critical tail $\bar t(N) \in (0, 1)$ such that, on a Pareto base with tail $t > \bar t(N)$, $N^* \geq N + 1$. The statement holds for the finite-draw contest with recall at every $k \geq 2$ and for the two-draw contest without recall.
\end{enumerate}
\end{proposition}

The two parts of Proposition~\ref{prop:optimal_N} share a single economic logic. With unlimited draws, each worker can explore the tail of $F$ as deeply as she wants, so adding workers only produces redundant above-threshold draws while forcing every worker to become less selective. Depth fully dominates breadth: only the best output matters, and a small number of very selective workers outperforms a large number of less selective ones, regardless of $F$. With finitely many draws, depth is capped: each worker has at most $k$ chances to explore the tail (with $k=2$ in the no-recall case), and no amount of selectivity can overcome that cap. When the tail is heavy enough---Pareto $t$ close to the infinite-mean boundary---the return to a marginal independent shot at the tail outweighs the cost of a less selective threshold, and larger fields become optimal. The same argument works with recall and without. The critical tail thickness $\bar t(N)$ depends on the contest form ($k$, recall versus no recall) and on $N$, but the qualitative pattern does not. As the deadline relaxes ($k \to \infty$), every tail threshold $\bar t(N, k) \to 1$ and part~(i) is recovered.

\section{Competing Designers}\label{sec:designers}

The individual contest embeds naturally within a hierarchy. Suppose $M \geq 2$ designers (schools, firms, teams) each employ $N$ workers. Each designer chooses a common search threshold $b$; her workers independently draw from $F$ at cost $c$ per draw until accepting a value above $b$. Designer $A$'s team output is $M_A = \max(X_1^A, \ldots, X_N^A)$. The designer with the highest team output wins a prize $\Omega > 0$; all others get zero. Designer $A$'s payoff is
\[
\Pi_A = \Omega \cdot \mathbf{1}\{M_A > \max_{j \neq A} M_j\} - \frac{Nc}{1-F(b_A)}.
\]
Differentiating the symmetric-equilibrium winning probability with respect to $b_A$ and applying the designer's first-order condition (see the Online Appendix) yields:

\begin{proposition}\label{prop:designer}
With $M$ competing designers of $N$ workers each, and $\Omega > cM(NM-1)/(M-1)$ so that the threshold below is interior, any symmetric equilibrium threshold satisfies
\begin{equation}\label{eq:designer_eq}
F(b_D) = 1 - \frac{cM(NM-1)}{\Omega(M-1)},
\end{equation}
which is distribution-free. Hence the symmetric equilibrium threshold, if one exists, is unique. At this profile, each designer's expected payoff is $\Omega(N-1)/[M(NM-1)] \geq 0$, and total search expenditure is $N\Omega(M-1)/(NM-1)$, giving dissipation ratio $N(M-1)/(NM-1) < 1$ for $N \geq 2$.
\end{proposition}

The distribution-free property arises for the same reason as in individual competition: the probability integral transform maps each worker's truncated draws to $\mathrm{Uniform}[0,1]$, making the strategic interaction independent of $F$. For $N = 1$ (each designer has one worker), the model reduces to the individual contest of Section~\ref{sec:equilibrium} with full dissipation, and $\Omega = W$ recovers equation~\eqref{eq:equilibrium}. As $N \to \infty$ the dissipation ratio converges to $(M-1)/M$, the standard Tullock dissipation ratio: with large teams each team's output concentrates near its upper tail and the meta-contest approaches a Tullock contest between $M$ players. For $N \geq 2$ with $\Omega = W$, $b_D < \lambda$: designers prescribe less search per worker than purely competitive workers because each designer internalizes the costs of all $N$ workers. For large $M$, the designer equilibrium matches the individual-contest equilibrium with prize $W = \Omega / M$.

\section{Conclusion}\label{sec:conclusion}

We have studied contests in which effort takes the form of sequential search. The organizing device is the move to quantile space: once draws are compared by rank, the equilibrium is distribution-free, and everything payoff-relevant depends on the observables, the number of players, the cost per draw, and the prize. One object, equilibrium rents, then organizes the comparative statics. Competition dissipates rents when search is unbounded and preserves them when depth is finite, and the response of standards to entry follows from how entry divides and erodes those rents. For design, the practical message is that the value distribution matters only through its upper tail: prizes and field sizes should be set by the chance of exceptional draws, not by average quality.

Several questions remain. Without recall, we do not know whether single-peakedness in the field size extends beyond two draws, and fixed-field comparative statics beyond three draws are open. Under recall, we know the threshold is single-peaked but not how the peak moves with the cost-prize ratio. Between the two conventions lies search with uncertain recall, connecting our two solved cases. The model ties optimal field size to tail thickness and search depth, a prediction that procurement and innovation data could confront.
\newpage
\appendix
\begin{center}
{\Large\bfseries Appendix}
\end{center}
\section{Proofs for Section~\ref{sec:equilibrium}}\label{app:multiprize}

\begin{proof}[Proof of Proposition~\ref{prop:multiprize}]
\emph{Part 1.} In a symmetric equilibrium with threshold $\lambda$, each player's final value has the truncated CDF $G(x)=(F(x)-F(\lambda))/(1-F(\lambda))$ for $x\geq \lambda$. A player who stops at $x$ has rank $j$ (exactly $j-1$ opponents beat her) with probability $\binom{N-1}{j-1}G(x)^{N-j}(1-G(x))^{j-1}$, so her expected prize from stopping at $x$ is $\Pi(x) = \sum_{j=1}^N W_j\binom{N-1}{j-1}G(x)^{N-j}(1-G(x))^{j-1}$, with $\Pi(\lambda) = W_N$ since $G(\lambda)=0$. The Bellman equation is $V = -c + F(\lambda)\cdot V + \int_\lambda^K \Pi(x)\,dF(x)$. Substituting $u=G(x)$ and using the beta integral $\int_0^1 \binom{N-1}{j-1}u^{N-j}(1-u)^{j-1}\,du = 1/N$, $\int_\lambda^K \Pi(x)\,dF(x) = (1-F(\lambda))\,\bar W$, so $V(1-F(\lambda)) = -c + (1-F(\lambda))\bar W$. Indifference at the threshold gives $V = \Pi(\lambda) = W_N$, and solving yields $1-F(\lambda)=c/(\bar W - W_N)$, interior by the assumption $c < \bar W - W_N$.

\emph{Part 2.} The equilibrium payoff is the value of stopping at the threshold: $V = W_N \geq 0$.

\emph{Part 3.} Each player's expected search cost is $c/(1-F(\lambda)) = \bar W - W_N$, so total expenditure is $N(\bar W - W_N) = \sum_j W_j - N W_N$ and the dissipation ratio is $1 - W_N/\bar W$.
\end{proof}

\section{Proofs for Section~\ref{sec:finite_recall}}\label{app:stationary}\label{app:finite_k_recall_eq}\label{app:rents}\label{app:finite_k_recall}

\begin{proof}[Proof of Lemma~\ref{lem:stationary}]
Let $\Phi$ denote the terminal-value CDF of an opponent induced by the stationary threshold strategy, and write $H(m) := \Phi(m)^{N-1}$ for a player's winning probability when she submits value $m$. Let $V_\tau(m)$ denote the player's value function with $\tau$ rounds remaining and running maximum $m$, with boundary $V_0(m) = W\,H(m)$ and recursion
\[
V_\tau(m) = \max\bigl\{W\,H(m),\; -c + \mathbb{E}[V_{\tau-1}(\max(m, X))]\bigr\}, \quad \tau \geq 1,
\]
where $X \sim F$. Define $\psi(m) := -c + \mathbb{E}\bigl[W\,H(\max(m, X))\bigr] - W\,H(m) = W\int_m^K [H(x) - H(m)]\,dF(x) - c$. Differentiating, $\psi'(m) = -W\,H'(m)\,[1 - F(m)] < 0$ for $m < K$, so $\psi$ is strictly decreasing. Let $\lambda^*$ be the unique root of $\psi$, so that $-c + \mathbb{E}\bigl[W\,H(\max(\lambda^*, X))\bigr] = W\,H(\lambda^*)$ and $\psi(m) > 0 \Leftrightarrow m < \lambda^*$.

We show by induction on $\tau$ that the optimal rule at round $\tau$ is to stop iff $m \geq \lambda^*$.

\emph{Base case} $\tau = 1$. Continuing means one more (terminal) draw, so stop and continue give the same payoff exactly when $\psi(m) = 0$, i.e., $m = \lambda^*$. Stopping is strictly better for $m > \lambda^*$ ($\psi(m) < 0$) and worse for $m < \lambda^*$ ($\psi(m) > 0$).

\emph{Inductive step.} Suppose the claim holds at $\tau - 1$. The round-$\tau$ continuation value is $-c + \mathbb{E}[V_{\tau-1}(\max(m, X))]$. Consider two cases. (i)~If $m \geq \lambda^*$, then $\max(m, X) \geq \lambda^*$ always, so the inductive hypothesis gives $V_{\tau-1}(\max(m, X)) = W\,H(\max(m, X))$, and the continuation value equals $W\,H(m) + \psi(m)$; since $m \geq \lambda^*$ implies $\psi(m) \leq 0$, continuing does not exceed stopping: stop. (ii)~If $m < \lambda^*$, using $V_{\tau-1} \geq W\,H$ pointwise, the continuation value is at least $W\,H(m) + \psi(m)$; since $m < \lambda^*$ implies $\psi(m) > 0$, continuing strictly exceeds stopping: continue.

The optimal round-$\tau$ policy is therefore the stationary threshold $\lambda^*$. Since the running maximum updates only on draws that exceed it, equivalently the player stops at round $t$ iff $x_t \geq \lambda^*$.
\end{proof}

\begin{proof}[Proof of Proposition~\ref{prop:finite_k_recall_eq}]
Our setup is the $E = 0$ case of \citet{Taylor1995}. By his Proposition~2 (p.~877), the research tournament has a unique equilibrium, and it is symmetric with every player using the same stationary threshold. In quantile space, this threshold is the acceptance quantile $a^*(N, k, c/W)$ that solves~\eqref{eq:gkR}. Existence and uniqueness of $a^*$ in $(0, 1)$ follow from the surplus form of the equilibrium equation: as shown in the proofs below, $c/W = g_k^R(N, a)$ is equivalent to $G(a, N) = 0$, where $\partial G/\partial a < 0$ on all of $(0, 1)$, $G(0^+, N) = W/N - c > 0$, and $G(1^-, N) = -ck < 0$. The equation therefore has exactly one solution for every $c/W \in (0, 1/N)$.
\end{proof}

\begin{proof}[Proof of Lemma~\ref{prop:surplus}]
Ties occur with probability zero, so by symmetry each player wins with probability $1/N$ and her equilibrium payoff is $V = W/N - K(a^*)$. A player draws until her first draw above the threshold, with at most $k$ attempts, so the number of draws is $\min(G, k)$ with $G$ geometric of parameter $1 - a^*$, and $K(a^*) = c\,\mathbb{E}[\min(G, k)] = c\,\frac{1 - (a^*)^k}{1 - a^*} = c\,\frac{1 - A}{1 - a^*}$. Substituting $c/W = g_k^R(N, a^*)$ from~\eqref{eq:gkR},
\[
K(a^*) \;=\; W (1 - A)\left[\frac{1 - A^N}{N(1 - A)} - A^{N-1}\right] \;=\; W\left[\frac{1 - A^N}{N} - (1 - A) A^{N-1}\right].
\]
Hence
\[
V \;=\; \frac{W}{N} - K(a^*) \;=\; W\left[\frac{A^N}{N} + (1 - A)A^{N-1}\right] \;=\; W A^{N-1}\Bigl(1 - \frac{N-1}{N}A\Bigr),
\]
which is~\eqref{eq:surplus}. Positivity follows from $A \in (0, 1)$ and $(N-1)A/N < 1$.
\end{proof}

\begin{proof}[Proof of Proposition~\ref{prop:corrected}]
Define, for $a \in (0, 1)$ and real $N \geq 2$,
\[
G(a, N) \;=\; \frac{W}{N} - K(a) - W R(a^k, N), \qquad R(A, N) \;=\; A^{N-1} - \frac{N-1}{N}A^{N},
\]
so that $G(a^*, N) = 0$ is the equilibrium condition~\eqref{eq:gkR} rearranged as in Lemma~\ref{prop:surplus}, and $WR(A, N) = V$ at equilibrium. The function $G$ is continuously differentiable in $(a, N)$. Its partials are as follows. First, $\frac{\partial R}{\partial N} = A^{N-1}\ln A - \frac{N-1}{N}A^{N}\ln A - \frac{A^N}{N^2} = (\ln A)\, R(A, N) - \frac{A^N}{N^2}$, so
\[
\frac{\partial G}{\partial N} \;=\; -\frac{W}{N^2} - W\frac{\partial R}{\partial N} \;=\; -\frac{W}{N^2}\bigl(1 - A^N\bigr) + V \ln(1/A),
\]
using $WR = V$ and $-\ln A = \ln(1/A)$. Second, $K'(a) = c\sum_{j=1}^{k-1} j a^{j-1} > 0$ and $\partial R/\partial A = (N-1)A^{N-2}(1 - A) > 0$, so $\frac{\partial G}{\partial a} = -K'(a) - W k a^{k-1} (N-1) A^{N-2}(1 - A) < 0$. The implicit function theorem gives $da^*/dN = -\,(\partial G/\partial N)/(\partial G/\partial a)$, which is~\eqref{eq:dadN}, with strictly positive denominator.
\end{proof}

\begin{proof}[Proof of Theorem~\ref{thm:singlepeak}]
Fix $k$ and $c/W \in (0, 1/2)$, and let $G(a, N)$ be as in the proof of Proposition~\ref{prop:corrected}. The rearrangement of~\eqref{eq:gkR} into $G = 0$ multiplies only by positive quantities, so the two equations have the same roots in $(0,1)$. For every real $N \in [2, W/c)$: $G(0^+, N) = W/N - c > 0$ since $R(0, N) = 0$ and $K(0^+) = c$; $G(1^-, N) = -ck$ since $R(1, N) = 1/N$ and $K(1^-) = ck$; and $\partial G/\partial a < 0$ on all of $(0, 1)$. Hence the equilibrium quantile $a^*(N)$ is the unique root of $G(\cdot, N)$, and by the implicit function theorem the path $N \mapsto a^*(N)$ is $C^1$ on $[2, W/c)$.

Let $\Phi(N, A) = R(A, N)\ln\tfrac{1}{A} - (1 - A^N)/N^2$, so that $\partial G/\partial N = W\Phi$, and, since $\partial G/\partial a < 0$, the sign of $da^*/dN$ equals the sign of $\psi(N) \equiv \Phi(N, A(N))$ along the path, with $A(N) = a^*(N)^k$.

\emph{Step 1: at any zero of $\Phi$, $\partial\Phi/\partial N < 0$.} Write $L = \ln(1/A)$. Using $\partial R/\partial N = (\ln A)R - A^N/N^2$ from the proof of Proposition~\ref{prop:corrected},
\[
\frac{\partial \Phi}{\partial N} \;=\; -L^2 R - \frac{2LA^N}{N^2} + \frac{2(1 - A^N)}{N^3}.
\]
At a zero of $\Phi$, $LR = (1 - A^N)/N^2$, so the first term equals $-L(1 - A^N)/N^2$ and
\[
\frac{\partial \Phi}{\partial N}\Big|_{\Phi = 0}
\;=\; \frac{1}{N^2}\Bigl[-L\bigl(1 + A^N\bigr) + \frac{2}{N}\bigl(1 - A^N\bigr)\Bigr]
\;=\; \frac{1}{N^3}\Bigl[-(1 + x)\ln\frac{1}{x} + 2(1 - x)\Bigr], \qquad x = A^N,
\]
using $NL = \ln(1/x)$. Negativity follows from the inequality $\ln(1/x) > 2(1 - x)/(1 + x)$ on $(0, 1)$: the difference $h(x) = -\ln x - 2(1 - x)/(1 + x)$ satisfies $h(1) = 0$ and $h'(x) = -(1 - x)^2/\bigl(x(1 + x)^2\bigr) < 0$.

\emph{Step 2: $\psi$ crosses zero at most once, downward.} At any zero of $\psi$, the numerator of~\eqref{eq:dadN} vanishes, so $da^*/dN = 0$ and hence $dA/dN = 0$; therefore $\psi'(N) = \partial\Phi/\partial N < 0$ by Step~1. A $C^1$ function with strictly negative derivative at every zero cannot vanish twice: after its first zero it is strictly negative, and a return to zero would require a nonnegative derivative at the returning zero. So exactly one of three configurations holds: $\psi > 0$ on the whole range (set $N^\dagger$ at the right endpoint), $\psi < 0$ on the whole range (set $N^\dagger = 2$), or $\psi > 0$ up to a unique zero $N^\dagger$ and $\psi < 0$ after.

\emph{Step 3: shape.} Since the sign of $da^*/dN$ equals the sign of $\psi$, the quantile is strictly increasing on $[2, N^\dagger)$ and strictly decreasing on $(N^\dagger, W/c)$. For the integer sequence, $a^*(n+1) - a^*(n) = \int_n^{n+1} (da^*/dN)\,dN$ is positive whenever $n + 1 \leq N^\dagger$, negative whenever $n \geq N^\dagger$, and of either sign only for the single increment with $n < N^\dagger < n + 1$. The increment signs therefore switch at most once, from positive to negative: the sequence rises to a peak and falls thereafter, with no second rise.
\end{proof}

\begin{proof}[Proof of Proposition~\ref{prop:finite_k_recall}]
\emph{(i) Selectivity for small cost.} Set $\epsilon = 1 - a$. Expanding $1 - a^m = m\epsilon - \tfrac{m(m-1)}{2}\epsilon^2 + O(\epsilon^3)$ at $m = kN$ and $m = k$ gives $(1 - a^{kN})/(N(1 - a^k)) = 1 - \tfrac{k(N-1)}{2}\epsilon + O(\epsilon^2)$, and subtracting $a^{k(N-1)} = 1 - k(N-1)\epsilon + O(\epsilon^2)$ and multiplying by $(1-a) = \epsilon$,
\[
g_k^R(N, a) \;=\; \frac{k(N-1)}{2}\,\epsilon^2 + O(\epsilon^3),
\qquad\text{so}\qquad
a^*(N, k, c/W) \;=\; 1 - \sqrt{\frac{2c}{kW(N-1)}} + O(c/W),
\]
strictly increasing in $N$ at leading order. The implicit function theorem, applied to $g_k^R(N, a) - c/W$, extends each pairwise inequality $a^*(N+1) > a^*(N)$ to an interval $(0, c_{N,k})$ with $c_{N,k} > 0$ determined by $g_k^R$ alone (hence $F$-free), and $\bar c(\bar N, k) := \min_{2 \leq N \leq \bar N} c_{N,k} > 0$ delivers the chain in the statement.

\emph{(ii) Eventual discouragement.} By Theorem~\ref{thm:singlepeak}, $a^*$ is strictly decreasing in $N$ on $(N^\dagger, W/c)$, so it suffices that the increasing piece not extend to the right endpoint. As $N \uparrow W/c$, $c/W \uparrow 1/N = g_k^R(N, 0)$, so $a^*(N, k, c/W) \to 0 < a^*(2, k, c/W)$, which is incompatible with $a^*$ increasing on the whole admissible range. Hence $N^\dagger < W/c$, and the sequence is strictly decreasing on the terminal range of admissible integers beyond $N^\dagger$.
\end{proof}

\section{Proofs for Section~\ref{sec:no_recall}}\label{app:nr_exist}\label{app:no_recall_eq}\label{app:finite_c}\label{app:nr_limits}

Throughout this appendix and the Online Appendix, payoffs are measured in units of the prize and $c$ denotes the cost-prize ratio $c/W \in [0, 1/N)$. Strategies are not restricted a priori: as noted after~\eqref{eq:nrvalues}, the Bellman recursion optimizes over all stopping rules, so a fixed point of the threshold best-response map is an equilibrium of the unrestricted game.

We prove Proposition~\ref{prop:nr_decline} first, since the existence argument uses it.

\begin{proof}[Proof of Proposition~\ref{prop:nr_decline}]
Monotonicity of values is by induction. First, $w_2 - w_1 = \int_0^1 \big[\max\{\Phi(u), w_1\} - \Phi(u)\big] du = \int_0^1 \max\{w_1 - \Phi(u),\, 0\}\,du \geq 0$. For the inductive step, $w_m \geq w_{m-1}$ implies $w_{m+1} - w_m = \int_0^1 \big[\max\{\Phi(u), w_m\} - \max\{\Phi(u), w_{m-1}\}\big]\,du \geq 0$, since $\max\{x, \cdot\}$ is nondecreasing. The threshold $\tau_m$ is nondecreasing in its target $w_{m-1}$ because $\Phi$ is nondecreasing, and the round-$r$ threshold uses target $w_{k-r}$, which falls in $r$; hence $a_1 \geq \cdots \geq a_{k-1}$.

For strictness, suppose $w_1 > 0$. Since $\Phi(0) = 0$ and $\Phi$ is continuous, the set $\{u : \Phi(u) < w_1\}$ has positive measure, so $w_2 - w_1 = \int \max\{w_1 - \Phi,\, 0\}\, du > 0$. Using $\max\{x, b\} - \max\{x, a\} \geq (b - a)\,\mathbf{1}\{x \leq a\}$ for $b \geq a$, $w_{m+1} - w_m \geq (w_m - w_{m-1})\cdot \mathrm{Leb}\{u : \Phi(u) \leq w_{m-1}\} > 0$, because $\{u : \Phi(u) \leq w_{m-1}\} \supseteq \{u : \Phi(u) < w_1\}$ has positive measure. Since $H$ is strictly increasing, so is $\Phi$, and $a_r = \Phi^{-1}(w_{k-r})$ is then strictly decreasing in $r$ because the targets $w_{k-r}$ are strictly decreasing in $r$.
\end{proof}

\begin{proof}[Proof of Proposition~\ref{prop:nr_exist}]
Let $D = \{a \in [0,1]^{k-1} : a_1 \geq a_2 \geq \cdots \geq a_{k-1}\}$, a convex compact set. Given $a \in D$, construct $H$ by~\eqref{eq:nrH}, the values by~\eqref{eq:nrvalues}, and the best-response vector $T(a) = (\tau_k(a), \ldots, \tau_2(a))$. By Proposition~\ref{prop:nr_decline}, $T(a) \in D$. The map $a \mapsto H$ is continuous in the sup norm; hence $\Phi$ and each $w_m$ are continuous in $a$. Moreover $w_m < 1$ for every $c \geq 0$: since $H$ is strictly increasing with $H(1) = 1$, $\Phi(u) < 1$ for $u < 1$, so $w_1 = \int_0^1 \Phi\,du - c < 1$; and if $w_{m-1} < 1$, then $\max\{\Phi(u), w_{m-1}\} < 1$ for every $u < 1$, so $w_m < 1$ by induction. Since $H$ is strictly increasing with $\Phi(0) = 0$ and $w_{m-1} < 1$, the generalized inverse $\tau_m = \Phi^{-1}(\max\{w_{m-1}, 0\})$ is continuous in $a$ as well. Brouwer's fixed point theorem applied to $T$ on $D$ delivers a profile with $a = T(a)$, which is a symmetric equilibrium.
\end{proof}

\medskip\noindent\emph{The two-draw contest.} Write $a = F(\lambda)$. The equilibrium quantile solves $g(N, a) = c/W$, where $g(N, a) = (1 + a^{2N-1})/(N(1+a)) - a^{2N-2}$. The function $g$ is continuous in $(a, c/W, N)$ and independent of $F$, so $a^*$ is a function of $(N, c/W)$ alone, with $g(N, 0) = 1/N$ and $g(N, 1) = -(N-1)/N < 0$.

\begin{lemma}\label{lem:gmono}
For every $N \geq 2$, $g(N, \cdot)$ is strictly decreasing on $(0, 1)$.
\end{lemma}

\begin{proof}
Write $m = 2N - 1 \geq 3$, so that $g(N, a) = (1 + a^m)/(N(1+a)) - a^{m-1}$ and $N = (m+1)/2$. Differentiating, $\frac{\partial g}{\partial a} = \frac{m a^{m-1}(1+a) - (1 + a^m)}{N(1+a)^2} - (m-1)\,a^{m-2}$, and $\partial g/\partial a < 0$ rearranges to
\[
m a^{m-1} + (m-1)\,a^m - 1 \;<\; \frac{m^2 - 1}{2}\,a^{m-2}(1+a)^2 .
\]
Since $a^m < a^{m-1}$ on $(0, 1)$, the left side is below $(2m-1)a^{m-1}$. Since $(1+a)^2 \geq 4a$, the right side is at least $2(m^2-1)a^{m-1}$. And $2(m^2 - 1) \geq 2m - 1$ for $m \geq 2$, so the inequality holds strictly on $(0, 1)$.
\end{proof}

\begin{proof}[Proof of Proposition~\ref{prop:no_recall_eq}]
By the intermediate value theorem and Lemma~\ref{lem:gmono}, $c/W = g(N, a)$ has a unique solution $a^*(N, c/W) \in (0, 1)$ for every $c/W \in [0, 1/N)$. By the derivation in the text, $c/W = g(N, a)$ is the round condition $\Phi(a) = w_1$, so the profile with round-1 threshold $\lambda^* = F^{-1}(a^*)$ is the unique symmetric equilibrium.
\end{proof}

\begin{proof}[Proof of Theorem~\ref{thm:nr_singlepeak}]
Multiplying $c/W = g(N, a)$ by $W(1+a)$ and collecting powers of $a$ shows that the equilibrium condition is equivalent to
\[
E(N, a) \;\equiv\; \frac{W}{N}\Bigl(1 - a^{2N-2}\bigl(N + (N-1)a\bigr)\Bigr) - (1+a)\,c \;=\; 0 .
\]
Both variable terms of $E$ are strictly decreasing in $a$, so $\partial E/\partial a < 0$ on $(0,1)$, the path $N \mapsto a^*(N)$ is $C^1$, and, multiplying $\partial E/\partial N$ by $N^2/W$, the sign of $da^*/dN$ equals the sign of
\[
\phi(N, a) \;\equiv\; 2N\ln\tfrac{1}{a}\; a^{2N-2}\bigl(N+(N-1)a\bigr)
- \Bigl(1 - a^{2N-2}\bigl(N+(N-1)a\bigr)\Bigr) - N\,a^{2N-2}(1+a).
\]
Since $E = 0$ is the payoff identity of the text, every path point has
$a^{2N-2}\bigl(N+(N-1)a\bigr) = NV/W \in (0, 1)$.

Let $\psi(N) = \phi(N, a^*(N))$ and consider a zero of $\psi$; there $da^*/dN = 0$, so
$\psi' = \partial\phi/\partial N$. Rearranging $\phi = 0$,
\begin{equation}\label{eq:nr2zero}
a^{2N-2}\,\Bigl[\bigl(N+(N-1)a\bigr)\Bigl(1 + 2N\ln\tfrac{1}{a}\Bigr) - N(1+a)\Bigr] \;=\; 1 .
\end{equation}
Differentiating $\phi$ at fixed $a$ and simplifying,
\[
\frac{\partial \phi}{\partial N}
\;=\; 4N\ln\tfrac{1}{a}\; a^{2N-2}\,\Bigl[(1+a) - \ln\tfrac{1}{a}\,\bigl(N+(N-1)a\bigr)\Bigr],
\]
negative if and only if $N\ln(1/a) > N(1+a)/\bigl(N+(N-1)a\bigr)$; and by~\eqref{eq:nr2zero},
$a^{2N-2}\bigl(N+(N-1)a\bigr) < 1$ if and only if
$2N\ln(1/a) > N(1+a)/\bigl(N+(N-1)a\bigr)$. The ratio $N(1+a)/\bigl(N+(N-1)a\bigr)$ is
increasing in $a$ and lies in $(1,\, 2N/(2N-1)]$.

Write $t = N\ln(1/a)$. A zero of $\psi$ has rent share below one, hence $t > 1/2$; since the
crossing is downward whenever $t$ exceeds the ratio, it remains to exclude
$t \in (1/2,\, 2N/(2N-1)]$. Suppose a zero has $t$ in that window. Using $2(N-1)t \geq 1$, $\bigl(N+(N-1)a\bigr)(1+2t) - N(1+a) = 2Nt - a\bigl[1-2(N-1)t\bigr] \geq 2Nt$, so~\eqref{eq:nr2zero} and $a^{2N-2} = e^{-2t(N-1)/N}$ give $e^{2t(N-1)/N} \geq 2Nt$. But
$e^{2t(N-1)/N} - 2Nt$ is convex in $t$ and negative at both endpoints of
$[1/2,\, 2N/(2N-1)]$: at $t = 1/2$ it equals $e^{(N-1)/N} - N$, negative since $e^{1/2} < 2$
and $e < 3$; at $t = 2N/(2N-1)$ it is at most $e^2(2N-1)/(2N+1) - 4N^2/(2N-1)$, negative
because $e^2 < 8$ and $2(2N-1)^2 \leq N^2(2N+1)$ for $N \geq 2$. So no such zero exists, every
zero of $\psi$ has $\psi' < 0$, and $\psi$ crosses zero at most once, and only downward, since a return to zero would require a nonnegative derivative at the returning zero: $a^*$ rises to
a single peak and falls. The integer statement follows by integrating $da^*/dN$ over unit
steps.
\end{proof}

\begin{proof}[Proof of Proposition~\ref{prop:finite_c}]
\emph{(i) Selectivity robust to small cost.} Consider first the case $c = 0$. Let $h_a(u)$ denote the CDF of an opponent's final score at the $u$-th quantile when all opponents use quantile threshold~$a$: $h_a(u) = au$ for $u < a$ and $h_a(u) = (1+a)u - a$ for $u \geq a$. With $U \sim \mathrm{Uniform}[0,1]$, equation~\eqref{eq:finite_threshold} at $c=0$ becomes $a^{2(N-1)} = \mathbb{E}[h_a(U)^{N-1}]$. Define $\Psi_N(a) := \mathbb{E}[h_a(U)^{N-1}] - a^{2(N-1)}$; the equilibrium quantile $a_N := a^*(N, 0)$ satisfies $\Psi_N(a_N) = 0$. Applying Lyapunov's inequality to the non-degenerate random variable $Z = h_{a_N}(U)$ with exponents $N > N - 1$: $\bigl(\mathbb{E}[Z^{N}]\bigr)^{1/N} > \bigl(\mathbb{E}[Z^{N-1}]\bigr)^{1/(N-1)} = a_N^2$. Raising to the $N$-th power: $\mathbb{E}[Z^N] > a_N^{2N}$, i.e., $\Psi_{N+1}(a_N) > 0$. Since $\Psi_{N+1}(1) = \tfrac{1}{N+1} - 1 < 0$, the intermediate value theorem yields $a_{N+1} \in (a_N, 1)$. This establishes $a^*(N+1, 0) > a^*(N, 0)$ for every $N \geq 2$.

The extension to $c > 0$ is by continuity. The map $(a, c/W) \mapsto g(N, a) - c/W$ is continuously differentiable with $\partial g/\partial a < 0$ (Lemma~\ref{lem:gmono}), so the implicit function theorem yields $a^*(N, c/W)$ continuous in $c/W$ in a neighborhood of $0$, with the neighborhood independent of $F$ because $g$ is. Hence for each $N \geq 2$ there exists $c_N > 0$ such that $a^*(N+1, c/W) > a^*(N, c/W)$ for all $c/W \in [0, c_N)$. For any fixed $\bar N \geq 2$, set $\bar c(\bar N) := \min_{2 \leq N \leq \bar N} c_N > 0$. For $c/W \in [0, \bar c(\bar N))$, the pairwise inequality $a^*(N+1) > a^*(N)$ holds simultaneously for every $N \in \{2, \ldots, \bar N\}$, giving the chain in the statement.

\emph{(ii) Eventual discouragement.} By Theorem~\ref{thm:nr_singlepeak}, $a^*$ is strictly decreasing in $N$ on $(N^\dagger, W/c)$, so it suffices that the increasing piece not extend to the right endpoint. Since $a^{2N-1}/(N(1+a)) \leq a^{2N-2}$, the equilibrium function satisfies $g(N, a) \leq 1/(N(1+a))$; hence $a^*(N, c/W) \leq (W - Nc)/(Nc) \to 0$ as $N \uparrow W/c$, which is incompatible with $a^*$ increasing on the whole admissible range, since $a^*(2, c/W) > 0$. Hence $N^\dagger < W/c$, and the sequence is strictly decreasing on the terminal range of admissible integers beyond $N^\dagger$.
\end{proof}

\medskip\noindent\emph{Large fields.} Notation as in Section~\ref{sec:nr_setup} and the preamble above: quantile thresholds $a_1 \geq \cdots \geq a_{k-1}$, gaps $g_r = 1 - a_r$, prefix products $P_r$, the distribution $H$ from~\eqref{eq:nrH}, heights $\theta_r = H(a_r)$, values $w_m$ from~\eqref{eq:nrvalues}, and round conditions $\theta_r^{\,N-1} = w_{k-r}$ at an interior equilibrium. Write $v_r = w_{k-r+1}$ for the value entering round $r$, so $v_1 \geq v_2 \geq \cdots \geq v_k$ by Proposition~\ref{prop:nr_decline} and the round conditions read $\theta_r^{\,N-1} = v_{r+1}$. Let $\sigma_0 = \sum_{l=1}^{k} P_l$ denote the slope of $H$ above $a_1$.

\begin{lemma}\label{lem:nrbasic}
In any symmetric equilibrium: (i) at $c = 0$, $v_1 = 1/N$ exactly, and at $c > 0$, $v_1 = 1/N - c\,\mathbb{E}[\#\text{draws}] \leq 1/N - c$; (ii) $H(u) \geq u^k$ for all $u$; (iii) $\theta_{k-1} \leq a_{k-1}^2$; (iv) $1 - \theta_1 = \sigma_0\,g_1$ and $\sigma_0 \leq k$; (v) $H$ is convex, so $H(u) \leq u$ and $\int_0^1 H^{N-1}\,du \leq 1/N$.
\end{lemma}

\begin{proof}
(i) Every player submits something, since round $k$ is forced; the highest submission wins and ties have probability zero because $H$ is continuous; so win probabilities sum to one and by symmetry each equals $1/N$. The expected payoff is the win probability minus expected search costs, and every player pays for at least one draw.
(ii) If all $k$ draws land below $u$, the submission is below $u$, whichever draw is kept.
(iii) Below $a_{k-1}$, $H$ is the single segment $H(u) = P_k u$, so $\theta_{k-1} = P_k a_{k-1}$, and $P_k = a_1\cdots a_{k-1} \leq a_{k-1}$ since every factor is at most one.
(iv) Above $a_1$ every kink of~\eqref{eq:nrH} is active, so $H$ is linear with slope $\sigma_0$ and $H(1) = 1$; and $\sigma_0 \leq k$ since each $P_l \leq 1$.
(v) The slope of $H$ increases each time $u$ crosses a threshold from below, so $H$ is convex; a convex function with $H(0) = 0$, $H(1) = 1$ lies below its chord.
\end{proof}

\begin{proof}[Proof of Theorem~\ref{thm:nr_largeN}]
Let $c = 0$; then every symmetric equilibrium is interior (all $v_{r+1} \geq v_k = \int_0^1 \Phi > 0$), and by Proposition~\ref{prop:nr_decline} the gaps satisfy $g_1 \leq \cdots \leq g_{k-1}$.

\emph{Part (i), ceiling on every gap.} By the round-$(k{-}1)$ condition and Lemma~\ref{lem:nrbasic}(ii),
\[
\theta_{k-1}^{\,N-1} \;=\; v_k \;=\; \int_0^1 H^{N-1}\,du \;\geq\; \int_0^1 u^{k(N-1)}\,du \;=\; \frac{1}{k(N-1)+1} .
\]
By Lemma~\ref{lem:nrbasic}(iii), $a_{k-1}^2 \geq \theta_{k-1} \geq \bigl(k(N-1)+1\bigr)^{-1/(N-1)}$, which gives the lower bound on $a_{k-1}$, and $a_r \geq a_{k-1}$ for every $r$. In gap form, $g_r \leq 1 - e^{-x} \leq x$ with $x = \ln(k(N-1)+1)/(2(N-1))$.

\emph{Part (i), floor on the top gap.} By the round-1 condition and Lemma~\ref{lem:nrbasic}(i), $\theta_1^{\,N-1} = v_2 \leq v_1 = 1/N$, so $\theta_1 \leq N^{-1/(N-1)}$. By Lemma~\ref{lem:nrbasic}(iv), $g_1 = (1 - \theta_1)/\sigma_0 \geq \bigl(1 - N^{-1/(N-1)}\bigr)/k$, which is the upper bound on $a_1$, hence on every $a_r$.

\emph{Part (ii).} It suffices that the largest gap at $MN$ players lie below the smallest gap at $N$ players:
\begin{equation}\label{eq:nrgoal}
\frac{\ln\bigl(k(MN-1)+1\bigr)}{2(MN-1)} \;<\; \frac{1 - N^{-1/(N-1)}}{k} .
\end{equation}
Bound the left side: $k(MN-1)+1 \leq kMN$ and $MN - 1 \geq M(N-1)$ give $\mathrm{LHS} \leq \ln(kMN)/(2M(N-1))$; and $t \mapsto \ln(ktN)/t$ is decreasing for $\ln(ktN) \geq 1$, so over $M \geq k$ the bound is largest at $M = k$: $\mathrm{LHS} \leq (\ln N + 2\ln k)/(2k(N-1))$. Bound the right side: with $x = \ln N/(N-1)$, the inequality $1 - e^{-x} \geq x - x^2/2$ gives $\mathrm{RHS} \geq \frac{\ln N}{k(N-1)}(1 - \delta)$ with $\delta = \ln N/(2(N-1))$. So~\eqref{eq:nrgoal} holds if $\tfrac12(\ln N + 2\ln k) < \ln N(1 - \delta)$, i.e.\ $\ln k < \ln N(\tfrac12 - \delta)$. Since $N \geq k^3$ gives $\ln k \leq \tfrac13\ln N$, it suffices that $\delta < \tfrac16$, i.e.\ $3\ln N < N - 1$, which holds for all $N \geq 8$ because $3\ln 8 < 7$ and $N - 1 - 3\ln N$ is increasing for $N \geq 3$; and $N \geq k^3 \geq 8$ since $k \geq 2$.
\end{proof}

\begin{proof}[Proof of Theorem~\ref{thm:nr_collapse}]
\emph{Part (i).} Suppose $a_{k-1} = 0$, which happens exactly when $v_k \leq 0$. Then $\max\{\Phi, v_k\} = \Phi$ pointwise, so $v_{k-1} = -c + \int\Phi = v_k \leq 0$, forcing $a_{k-2} = 0$; the same step walks up the profile, so every threshold is zero. But at the all-zero profile $H(u) = u$, so $v_k = 1/N - c > 0$ by the maintained assumption $Nc < 1$, contradicting $v_k \leq 0$. Hence $a_{k-1} > 0$, and since thresholds decline across rounds, every threshold is positive.

\emph{Part (ii).} By part (i), $v_k > 0$. Bound $v_k = -c + \int_0^1 H^{N-1}\,du$ from above in two pieces. Below $a_1$, $H \leq \theta_1$, so $\int_0^{a_1}H^{N-1} \leq a_1\theta_1^{\,N-1} = a_1 v_2$ by the round-1 condition, and $v_2 \leq v_1 \leq 1/N - c$ by Lemma~\ref{lem:nrbasic}(i). Above $a_1$, $H$ is linear with slope $\sigma_0$ and reaches one, so $\int_{a_1}^1 H^{N-1} = (1-\theta_1^{\,N})/(N\sigma_0) \leq 1/(N\sigma_0) \leq 1/(N(1+a_1))$, using $\sigma_0 \geq P_1 + P_2 = 1 + a_1$. Therefore
\[
0 \;<\; N v_k \;\leq\; -Nc + a_1(1 - Nc) + \frac{1}{1+a_1} .
\]
Writing $\delta = 1 - Nc$ and multiplying by $1 + a_1$, the right side rearranges to $\delta(1+a_1)^2 - a_1 > 0$, which is the stated ceiling. As $Nc \uparrow 1$, $a_1 \leq \delta(1+a_1)^2 \leq 4\delta \to 0$, and every threshold lies below $a_1$.
\end{proof}

\section{Proofs for Section~\ref{sec:planner}}\label{app:tail_driven}\label{app:N_decreasing}\label{app:heavy_tail_no_recall}\label{app:heavy_tail_recall}

\begin{proof}[Proof of Propositions~\ref{prop:location_inv} and~\ref{prop:tail_mono}]
In all three cases, the planner's welfare admits the common form
\[
\mathrm{SW}(a; F) \;=\; \int_{q_0(a)}^1 F^{-1}(q)\,\psi(q; N, a)\,dq \;-\; C(a),
\]
where $\psi(q; N, a)$ is the density of the maximum submitted quantile, $q_0(a)$ is the lower edge of the submission support, and $C(a)$ is an $F$-independent cost term. Specifically: \emph{unlimited draws:} $q_0(a) = a$, $\psi^\infty(q; N, a) = N(q - a)^{N-1}/(1-a)^N$ on $[a, 1]$, where $a = 1 - Nc/W$ is the equilibrium acceptance quantile, and $C(a) = W(a) = Nc/(1-a)$ (with $\alpha$ irrelevant, as shown in Section~\ref{sec:planner_obj}); \emph{finite-draw with recall:} $q_0(a) = 0$, $\psi^R(q; N, a, k) = N H_k^R(q; a)^{N-1} h_k^R(q; a)$, with $H_k^R(q; a) = q^k$ on $[0, a]$ and $H_k^R(q; a) = a^k + (q-a)(1-a^k)/(1-a)$ on $[a, 1]$, and cost term $C(a) = (1-\alpha)\,c/g_k^R(N, a) + \alpha N c(1-a^k)/(1-a)$; \emph{no recall ($k=2$):} $q_0(a) = 0$, $\psi(q; N, a) = N H_a(q)^{N-1} H_a'(q)$ with $H_a(q) = aq$ on $[0, a]$ and $H_a(q) = q(1+a) - a$ on $[a, 1]$, and cost term $C(a) = (1-\alpha)\,c/g(N, a) + \alpha N c(1+a)$.

In each case $\int_{q_0(a)}^1 \psi(q; N, a)\, dq = 1$, and $W^*$ is determined by the optimal $a^*$ through the case-specific prize-to-threshold mapping: $W^* = Nc/(1-a^*)$ for unlimited draws, $W^* = c/g_k^R(N, a^*)$ for finite-draw recall, and $W^* = c/g(N, a^*)$ for no recall. Each mapping is strictly increasing in $a^*$ ($g_k^R$ and $g$ are strictly decreasing in $a$, and $Nc/(1-a)$ is increasing in $a$).

\emph{Proposition~\ref{prop:location_inv} (location invariance).} Under $F_\theta^{-1}(q) = F_0^{-1}(q) + \theta$,
\[
\int_{q_0(a)}^1 F_\theta^{-1}(q)\,\psi(q;N,a)\,dq \;=\; \int_{q_0(a)}^1 F_0^{-1}(q)\,\psi(q;N,a)\,dq \;+\; \theta,
\]
since $\int \psi = 1$. The $\theta$ term is constant in $a$ and drops out of the FOC, so $a^*$ is independent of $\theta$; hence $W^*(F_\theta) = W^*(F_0)$.

\emph{Proposition~\ref{prop:tail_mono} (tail monotonicity).} By location invariance, $W^*(F_M) = W^*(\tilde F_M)$ where $\tilde F_M^{-1}(q) = \min(F_0^{-1}(q), M)$ (drop the mean-preserving shift $\theta(M)$: it does not affect $W^*$). It therefore suffices to show that $W^*(\tilde F_M)$ is strictly increasing in $M$. With $q_M = F_0(M) > q^*$, $\partial \tilde F_M^{-1}/\partial M = \mathbf{1}\{q > q_M\}$, so $\partial^2 \mathrm{SW}/\partial a\,\partial M = \int_{q_M}^1 (\partial \psi/\partial a)\,dq = (\partial/\partial a)\int_{q_M}^1 \psi\,dq$. Since $\int_{q_M}^1 \psi\, dq = 1 - H(q_M; a, \ldots)^N$ where $H$ is the (case-specific) submitted-quantile CDF, the cross-partial reduces to $-N H(q_M)^{N-1} \partial H(q_M)/\partial a$ in every case. We verify $\partial H(q_M)/\partial a < 0$ for $q_M > a$ in each case:

\emph{Unlimited draws:} The conditional CDF is $H^\infty(q; a) = (q-a)/(1-a)$ on $[a, 1]$, so $\partial H^\infty/\partial a = -(1-q)/(1-a)^2 < 0$.

\emph{Finite-draw with recall:} On $[a, 1]$, $H_k^R(q_M; a) = a^k + (q_M - a)(1-a^k)/(1-a)$. Writing $S(a) = (1-a^k)/(1-a) = 1 + a + \cdots + a^{k-1}$, $\partial H_k^R(q_M)/\partial a = ka^{k-1} - S(a) + (q_M - a)S'(a)$. For fixed $a$, this expression is linear in $q_M$ on $[a, 1]$, so its values at the two endpoints pin it down. At $q_M = a$, $\partial H_k^R/\partial a = ka^{k-1} - S(a) = ka^{k-1} - (1 + a + \cdots + a^{k-1}) < 0$ (each of the $k$ terms $1, a, \ldots, a^{k-1}$ strictly exceeds $a^{k-1}$ except the last, for $a < 1$). At $q_M = 1$, $H_k^R(1; a) = 1$ for all $a$, so $\partial H_k^R(1)/\partial a = 0$. A linear function that goes from a strictly negative value at $q_M = a$ to zero at $q_M = 1$ is strictly negative on $(a, 1)$.

\emph{No recall ($k=2$):} $H_a(q_M) = (1+a)q_M - a$, $\partial H_a(q_M)/\partial a = q_M - 1 < 0$.

So $\partial^2 \mathrm{SW}/\partial a\,\partial M > 0$ whenever $a < q_M$. For $a > q_M$ the cap sits below the acceptance threshold: in the unlimited-draw and recall contests the branch of $H$ below the threshold does not involve $a$, so the cross-partial is zero there; in the two-draw no-recall contest feasible thresholds satisfy $a \leq \bar a(N) < q_M$ on the stated range, so the case does not arise. Three steps conclude.

\emph{Step 1: on the stated range, every maximizer lies strictly below $q_M$.} For the no-recall contest this is the feasibility bound. For the other two contests, suppose a maximizer $a$ at some $M$ in the range had $a > q^*(F_0)$. Adding the optimality inequalities $\mathrm{SW}(a; M) \geq \mathrm{SW}(q^*(F_0); M)$ and $\mathrm{SW}(q^*(F_0); F_0) \geq \mathrm{SW}(a; F_0)$, and writing the difference between the $F_0$ problem and the $M$ problem as the integral of $\partial \mathrm{SW}/\partial m$ over the caps $m \in (M, F_0^{-1}(1^-))$, gives
\[
\int_{q^*(F_0)}^{a}\!\int_{M}^{F_0^{-1}(1^-)} \frac{\partial^2 \mathrm{SW}}{\partial a\,\partial m}\;dm\,da \;\leq\; 0 .
\]
The integrand is nonnegative and strictly positive wherever the threshold lies below $q_m$, which holds on a set of positive measure because $q_m \to 1$: a contradiction. So every maximizer is at most $q^*(F_0) < q_M$.

\emph{Step 2: maximizers are nondecreasing in $M$.} For $M < M'$ in the range with maximizers $a$ and $a'$, if $a' < a$ the same two-point sum gives $\int_{a'}^{a}\int_{M}^{M'} \partial^2 \mathrm{SW}/\partial a\,\partial m \leq 0$, contradicting strict positivity of the integrand on the rectangle, which lies below $q_M$ by Step~1.

\emph{Step 3: the increase is strict.} At an interior maximizer, $\partial \mathrm{SW}/\partial a\,(a^*(M); M) = 0$, and integrating the strictly positive cross-partial over the caps in $(M, M')$ gives $\partial \mathrm{SW}/\partial a\,(a^*(M); M') > 0$, so $a^*(M)$ is not optimal at $M'$; with Step~2, $a^*(M') > a^*(M)$. Since the prize-to-threshold mapping is strictly increasing in $a^*$, $W^*(F_M)$ is strictly increasing in $M$.
\end{proof}

\begin{proof}[Proof of Proposition~\ref{prop:optimal_N}, part (i)]
Fix $N\geq 1$. Let $p^* = 1-F(b^*_{N+1})$ denote the optimal acceptance probability in the $(N{+}1)$-worker problem, so that total cost is $(N{+}1)c/p^*$.

\emph{Step 1 (Equal-cost comparison).} We show that $N$ more-selective workers produce a higher expected maximum than $N{+}1$ less-selective workers at the same total cost. The $(N{+}1)$-worker optimum uses acceptance probability~$p^*$ at total cost $(N{+}1)c/p^*$. Consider instead $N$ workers with acceptance probability $\tilde p = Np^*/(N{+}1)$. Since $\tilde p < p^*$, these workers are more selective, and their total cost $Nc/\tilde p = (N{+}1)c/p^*$ is identical.

\emph{Step 2 (Quantile representation).} Each worker~$i$ independently accepts a draw~$X_i$ above the threshold. By the probability integral transform, the quantile $q_i = F(X_i)$ is uniform on $[1{-}p,\,1]$. Since $F^{-1}$ is increasing, the best output is $\max_i X_i = F^{-1}(\max_i q_i)$, so it suffices to show that the highest quantile $Q_N \equiv \max_i q_i$ under the $N$-worker allocation first-order stochastically dominates $Q_{N+1}$ under the $(N{+}1)$-worker allocation.

The maximum of $N{+}1$ independent uniforms on $[1{-}p^*,\,1]$ has CDF $((q{-}(1{-}p^*))/p^*)^{N+1}$, and the maximum of $N$ independent uniforms on $[1{-}\tilde p,\,1]$ has CDF $((q{-}(1{-}\tilde p))/\tilde p)^{N}$. To compare these on a common support, set $u = (q - (1{-}p^*))/p^* \in [0,1]$. Then $Q_{N+1}$ has CDF $u^{N+1}$ on $[0,1]$, while $Q_N$ has CDF $[(u(N{+}1){-}1)/N]^N$ on $[1/(N{+}1),\,1]$ (and $0$ below $1/(N{+}1)$, since the $N$ workers' higher selectivity excludes the bottom of the support).

\emph{Step 3 (FOSD).} We claim $Q_N$ first-order stochastically dominates~$Q_{N+1}$:
\begin{equation}\label{eq:FOSD}
u^{N+1} \;\geq\; \Bigl[\frac{u(N{+}1)-1}{N}\Bigr]^N \qquad \text{for all } u\in\bigl[\tfrac{1}{N{+}1},\,1\bigr],
\end{equation}
with equality only at $u=1$. Setting $t = (u(N{+}1)-1)/N \in [0,1]$, the inequality becomes $((Nt+1)/(N{+}1))^{N+1}\geq t^N$. The left side is the $(N{+}1)$-th power of $(Nt+1)/(N{+}1)$, which is the arithmetic mean of $N$ copies of~$t$ and one copy of~$1$. Since the arithmetic mean of non-negative numbers is at least their geometric mean, $(Nt + 1)/(N+1) \geq (t^N \cdot 1)^{1/(N+1)} = t^{N/(N+1)}$; raising to the $(N{+}1)$-th power yields~\eqref{eq:FOSD}.

\emph{Step 4.} Since $F^{-1}$ is nondecreasing and $Q_N$ strictly FOSD $Q_{N+1}$, $\mathbb{E}[F^{-1}(Q_N)] > \mathbb{E}[F^{-1}(Q_{N+1})]$. The two setups have the same total cost, so $\mathrm{SW}^*(N) \geq \mathrm{SW}(N, \tilde p) > \mathrm{SW}^*(N+1)$, where the first inequality is by optimality of~$b^*(N)$.
\end{proof}

\begin{proof}[Proof of Proposition~\ref{prop:optimal_N}, part (ii), two-draw no recall]
Let $a \in [0,\bar a(N)]$, where $\bar a(N)$ is the positive root of the zero-profit condition $g(N,a) = \frac{1+a^{2N-1}}{N(1+a)} - a^{2N-2} = 0$. The welfare function is
\[
\mathrm{SW}(N,a,t;\alpha) \;=\; \int_0^1 (1-q)^{-t}\, \psi(q;N,a)\, dq \;-\; (1-\alpha)\,\frac{c}{g(N,a)} \;-\; \alpha\cdot Nc(1+a),
\]
and $\mathrm{SW}^*(N,t;\alpha) = \max_{a \in [0,\bar a(N)]} \mathrm{SW}(N,a,t;\alpha)$. Both cost terms are bounded as $t \to 1^-$, so the tail asymptotic below does not depend on $\alpha$.

\emph{Step 1 (Asymptotics as $t \to 1^-$).} Recall the CDF of each worker's submitted score: $H_a(q) = aq$ for $q \leq a$ and $H_a(q) = q(1+a)-a$ for $q > a$. The density of the overall maximum is $\psi(q;N,a) = NH_a(q)^{N-1}H_a'(q)$. For $q > a$, $\psi(q;N,a) = N(q(1+a)-a)^{N-1}(1+a)$, so $\psi(1^-;N,a) = N(1+a)$.

Substituting $u = q(1+a)-a$ on $[a,1]$, with $1-q = (1+a-u)/(1+a)$:
\[
\int_a^1 (1-q)^{-t}\,\psi(q;N,a)\,dq \;=\; (1+a)^t \int_{a^2}^1 (1-u)^{-t}\,Nu^{N-1}\,du.
\]
Since $\int_0^1 (1-u)^{-t}\,Nu^{N-1}\,du = NB(N,1-t) \sim N/(1-t)$ as $t \to 1^-$ (where $B$ is the Beta function), and the integral over $[0,a]$ is bounded for any fixed $a < 1$, we obtain
\[
\int_0^1 (1-q)^{-t}\,\psi(q;N,a)\,dq \;\sim\; \frac{N(1+a)}{1-t} \quad \text{as } t \to 1^-.
\]
Since both cost terms are bounded for any fixed $a < \bar a(N)$ and any $\alpha \in [0, 1]$:
\begin{equation}\label{eq:asymp_fixed_a}
\lim_{t \to 1^-} (1-t)\,\mathrm{SW}(N,a,t;\alpha) \;=\; N(1+a), \qquad \text{for each fixed } a < \bar a(N).
\end{equation}
Because $\mathrm{SW}^*(N,t;\alpha) \geq \mathrm{SW}(N,a,t;\alpha)$ for every $a < \bar a(N)$, taking $a \uparrow \bar a(N)$ gives $\liminf_{t \to 1^-} (1-t)\,\mathrm{SW}^*(N,t;\alpha) \geq N(1+\bar a(N))$. The matching upper bound is uniform over $a \in [0, \bar a(N)]$: on $[a, 1]$ the substitution above gives at most $(1+\bar a(N))^t\, N B(N, 1-t)$, on $[0, a]$ the integrand is at most $(1 - \bar a(N))^{-t}$, which stays bounded as $t \to 1^-$, and the cost terms are nonnegative; hence $\limsup_{t \to 1^-} (1-t)\,\mathrm{SW}^*(N,t;\alpha) \leq N(1+\bar a(N))$. Therefore
\begin{equation}\label{eq:asymp}
\lim_{t \to 1^-} (1-t)\,\mathrm{SW}^*(N,t;\alpha) \;=\; N(1+\bar a(N)),
\end{equation}
independent of $\alpha$.

\emph{Step 2 ($\bar a(N)$ is increasing in $N$).} The admissibility boundary $\bar a(N)$ is the zero-cost equilibrium threshold $a^*(N, 0)$. By Proposition~\ref{prop:finite_c} Part~(i), $a^*(N, c/W)$ is strictly increasing in $N$ on some range $[0, \bar c(\bar N))$, and in particular at $c/W = 0$: $\bar a(N+1) = a^*(N+1, 0) > a^*(N, 0) = \bar a(N)$.

\emph{Step 3 (Conclusion).} By Step~2 applied at every field size, $\bar a(\cdot)$ is strictly increasing, so the limit coefficient $m(1 + \bar a(m))$ in~\eqref{eq:asymp} is strictly increasing in $m$. Fix $N \geq 2$. For each $m \leq N$, \eqref{eq:asymp} gives $(1-t)\bigl[\mathrm{SW}^*(N+1,t) - \mathrm{SW}^*(m,t)\bigr] \to (N+1)\bigl(1+\bar a(N+1)\bigr) - m\bigl(1+\bar a(m)\bigr) > 0$, so there is $t_m < 1$ with $\mathrm{SW}^*(N+1,t) > \mathrm{SW}^*(m,t)$ for all $t \in (t_m, 1)$. With $\bar t(N) = \max_{2 \leq m \leq N} t_m$, every $t > \bar t(N)$ has the field size $N+1$ dominating every $m \leq N$, hence $N^* \geq N + 1$.
\end{proof}

\begin{proof}[Proof of Proposition~\ref{prop:optimal_N}, part (ii), finite-draw recall]
Under recall, the admissibility boundary is $\bar a^R(N, k) = 1$, since $g_k^R(N, 1) = 0$. The tail asymptotic of Step~1 of the preceding proof uses only the upper branch of the submitted-quantile CDF. The upper branch of $H_k^R$ is linear from $a^k$ to $1$ with slope $(1-a^k)/(1-a) \leq k$, which makes the limit below uniform over $a$; substituting $u = H_k^R(q)$ in the tail integral yields $\int_{a^k}^1 (1-u)^{-t} N u^{N-1} du \cdot [(1-a^k)/(1-a)]^t$, so
\[
\lim_{t \to 1^-} (1-t)\,\mathrm{SW}^*(N, k, t) = N \cdot \frac{1 - (\bar a^R)^k}{1 - \bar a^R}\bigg|_{\bar a^R = 1} = Nk,
\]
using $(1-a^k)/(1-a) = 1 + a + \cdots + a^{k-1} \to k$ as $a \to 1$. The limit coefficient $mk$ is strictly increasing in the field size $m$. Fix $N \geq 2$: for each $m \leq N$, $(1-t)\bigl[\mathrm{SW}^*(N+1, k, t) - \mathrm{SW}^*(m, k, t)\bigr] \to (N+1-m)k > 0$, so there is $t_m < 1$ with $\mathrm{SW}^*(N+1, k, t) > \mathrm{SW}^*(m, k, t)$ for all $t \in (t_m, 1)$. With $\bar t(N, k) = \max_{2 \leq m \leq N} t_m$, every $t > \bar t(N, k)$ has $N+1$ dominating every $m \leq N$, hence $N^* \geq N + 1$.
\end{proof}

\setstretch{1}

\end{document}


\maketitle

\noindent This online appendix contains the proofs for the three-draw contest without recall (Section~5.3 of the paper) and the derivation of the designer first-order condition (Section~7). Numbered cross-references without the OA prefix point to the paper. Throughout the three-draw proofs, payoffs are measured in units of the prize and $c$ denotes the cost-prize ratio $c/W \in [0, 1/N)$.

\section{Proofs for Section~\ref{sec:nr_k3}: the Three-Draw Contest}\label{app:nr_k3}

The reduction behind Theorem~\ref{thm:nr_k3} is the following proposition. Here $\alpha = H(a_1)$ and $\beta = H(a_2)$ are the two threshold heights.

\begin{proposition}\label{prop:nr_k3red}
In any interior symmetric equilibrium of the three-draw contest, at any $c/W \in [0, 1/N)$:
\begin{enumerate}
\item[(i)] the threshold heights satisfy the link $\alpha^{N-1} = \beta^{N-1}\bigl(1 + \tfrac{N-1}{N}\,a_2\bigr)$;
\item[(ii)] the round-1 threshold is determined by the round-2 threshold in closed form,
\[
a_1 \;=\; B_N(a_2) \;\equiv\; \frac{a_2 + a_2^2\,\varphi_N(a_2)}{1 + a_2},
\qquad
\varphi_N(a_2) \;\equiv\; \Bigl(1 + \frac{N-1}{N}\,a_2\Bigr)^{1/(N-1)} ;
\]
\item[(iii)] substituting $a_1 = B_N(a_2)$ into the heights $\alpha$, $\beta$ and the slopes $s_1 = a_1 a_2$, $s_2 = a_1(1+a_2)$, $s_3 = 1 + a_1 + a_1 a_2$ makes each a function of $a_2$ alone, and the interior symmetric equilibria correspond exactly to the roots of
\begin{equation}\label{eq:nrPsi}
\Psi_N(a_2) \;=\; \frac{1}{N}\Bigl[\frac{\beta^N}{s_1} + \frac{\alpha^N - \beta^N}{s_2}
+ \frac{1 - \alpha^N}{s_3}\Bigr] \;-\; \frac{c}{W} \;-\; \beta^{N-1}
\end{equation}
on $(0, \bar a_N)$, where $\bar a_N$ solves $B_N(\bar a_N) = 1$.
\end{enumerate}
\end{proposition}

The lemmas below prove the proposition in order, the link, the branch, and the converse; the proofs of uniqueness and selectivity, which together prove Theorem~\ref{thm:nr_k3}, follow them.

Notation: write $a = a_1$ and $b = a_2$ for the two thresholds, $\alpha = H(a)$ and $\beta = H(b)$ for their heights, and $c$ for the cost-prize ratio, with $c \in [0, 1/N)$. For $k = 3$ the distribution~\eqref{eq:nrH} is
\begin{equation}\label{eq:nrH3}
H(u) \;=\;
\begin{cases}
abu & u \in [0, b],\\[2pt]
a\bigl((1+b)u - b\bigr) & u \in [b, a],\\[2pt]
(1 + a + ab)\,u - a - ab & u \in [a, 1],
\end{cases}
\end{equation}
with $\beta = ab^2$ and $\alpha = a\bigl((1+b)a - b\bigr)$, and the values~\eqref{eq:nrvalues} specialize to
\begin{equation}\label{eq:nrvalues3}
w_1 = -c + \int_0^1 \Phi(u)\,du,
\qquad
w_2 = -c + \int_0^1 \max\{\Phi(u), w_1\}\,du .
\end{equation}
An equilibrium is a pair reproducing itself, $\Phi(b) = w_1$ and $\Phi(a) = w_2$; by Proposition~\ref{prop:nr_decline}, $w_2 \geq w_1$ and hence $a \geq b$.

\begin{lemma}\label{lem:nrinterior}
Every symmetric equilibrium of the three-draw contest with $c \in [0, 1/N)$ is interior: $w_1 > 0$ and $w_2 < 1$, hence $0 < b \leq a < 1$ and both round conditions hold with equality.
\end{lemma}

\begin{proof}
If $w_1 \leq 0$, then $\max\{\Phi, w_1\} = \Phi$ pointwise, so $w_2 = w_1$, both thresholds are zero, $H(u) = u$, and $w_1 = 1/N - c > 0$, a contradiction. And $w_1, w_2 < 1$ because $H$ is strictly increasing with $H(1) = 1$, so $\Phi(u) < 1$ for $u < 1$. Since $\Phi$ is continuous and strictly increasing with $\Phi(0) = 0$ and $\Phi(1) = 1$, the conditions $\Phi(b) = w_1 \in (0,1)$ and $\Phi(a) = w_2 \in (0,1)$ place both thresholds strictly inside $(0,1)$.
\end{proof}

\subsection*{The link and the scalar reduction}

An interior equilibrium is a pair $(a, b)$ meeting the two indifference conditions $\Phi(b) = w_1$
and $\Phi(a) = w_2$. We rewrite this pair as an equivalent one: a relation between the two
thresholds that solves in closed form for $a$, and a single equation in $b$. The relation comes from
subtracting one indifference condition from the other.

\begin{lemma}\label{lem:nrlink}
In any interior symmetric equilibrium, at any $c$,
\begin{equation}\label{eq:nrlink}
\alpha^{N-1} \;=\; \beta^{N-1}\Bigl(1 + \frac{N-1}{N}\,b\Bigr).
\end{equation}
\end{lemma}

\begin{proof}
Subtracting the two equations in~\eqref{eq:nrvalues3} cancels the cost:
$w_2 - w_1 = \int_0^1 \max\{w_1 - \Phi(u), 0\}\dd u = \int_0^b (w_1 - \Phi)\dd u$, since
$\Phi(u) \geq w_1$ exactly on $u \geq b$. On $[0, b]$, $H = abu$, so with
$w_1 = \beta^{N-1}$,
\[
\int_0^b\bigl(\beta^{N-1} - (abu)^{N-1}\bigr)\dd u
\;=\; b\,\beta^{N-1} - \frac{b\,\beta^{N-1}}{N}
\;=\; \frac{N-1}{N}\,b\,\beta^{N-1}.
\]
Substituting the two indifference conditions gives~\eqref{eq:nrlink}.
\end{proof}

We call~\eqref{eq:nrlink} the \emph{link}: it ties the two thresholds $a$ and $b$ together, through
their heights $\alpha$ and $\beta$.

\begin{corollary}\label{cor:nrbranch}
In any interior symmetric equilibrium, at any $c$, the round-1 threshold is
\begin{equation}\label{eq:nrbranch}
a \;=\; B_N(b) \;:=\; \frac{b + b^2\varphi_N(b)}{1+b},
\qquad
\varphi_N(b) \;:=\; \Bigl(1 + \frac{N-1}{N}\,b\Bigr)^{1/(N-1)} .
\end{equation}
\end{corollary}
\begin{proof}
Substituting $\alpha$ and $\beta$ into~\eqref{eq:nrlink} and cancelling $a^{N-1}$ we obtain
\[
\Bigl(\frac{(1+b)a - b}{b^2}\Bigr)^{N-1} = 1 + \frac{N-1}{N}b.
\]
Solving for $a$ gives the closed form in the corollary.
\end{proof}

Note that $a = B_N(b)$ can be written as
$\alpha = \beta\,\varphi_N(b)$. The map $B_N$ is strictly increasing, and
$B_N(b) - b = b^2(\varphi_N(b) - 1)/(1+b) \geq 0$ since $\varphi_N \geq 1$, so $B_N(b) \geq b$
throughout the domain.

Substituting $a = B_N(b)$ into the round-2 indifference $\Phi(b) = w_1$ eliminates $a$ and reduces
the equilibrium to a single equation in $b$. That condition reads
$\beta^{N-1} = \int_0^1 H^{N-1}\dd u - c$; with $\alpha$, $\beta$, and the slopes $s_1 = ab$,
$s_2 = a(1+b)$, $s_3 = 1 + a + ab$ all taken at $a = B_N(b)$, hence functions of $b$ alone,
integrating $H^{N-1}$ over the three linear pieces of $H$ writes it as $\Psi_N(b) = 0$, where
\begin{equation}\label{eq:nrPsiA}
\Psi_N(b) \;=\; \frac{1}{N}\Bigl[\frac{\beta^N}{s_1} + \frac{\alpha^N - \beta^N}{s_2}
+ \frac{1 - \alpha^N}{s_3}\Bigr] - c - \beta^{N-1}.
\end{equation}
Every interior equilibrium is therefore a root of $\Psi_N$ on $(0, \bar b_N)$, where
$B_N(\bar b_N) = 1$ (equivalently $\bar b_N^2\,\varphi_N(\bar b_N) = 1$). The converse holds as well.

\begin{lemma}\label{lem:nrconverse}
Fix $b \in (0, \bar b_N)$ and set $a = B_N(b)$. If $\Psi_N(b) = 0$, then $(a, b)$ is an interior
symmetric equilibrium. The interior equilibria are exactly the pairs $(B_N(b), b)$ with
$\Psi_N(b) = 0$.
\end{lemma}

\begin{proof}
Since $b \in (0, \bar b_N)$ and $B_N$ is increasing with $B_N(\bar b_N) = 1$, we have
$0 < b \leq a < 1$, so $(a, b)$ is interior. By construction
$\Psi_N(b) = \int_0^1 H^{N-1}\dd u - c - \beta^{N-1} = w_1 - \Phi(b)$, so $\Psi_N(b) = 0$ is the
round-2 indifference $\Phi(b) = w_1$; it gives $w_1 = \beta^{N-1}$ and fixes the round-2 acceptance
region as $[b, 1]$. From the definition of $w_2$,
\[
w_2 - w_1 \;=\; \int_0^b (w_1 - \Phi)\dd u \;=\; \tfrac{N-1}{N}\,b\,\beta^{N-1},
\]
so $w_2 = \beta^{N-1}(1 + \tfrac{N-1}{N}b)$. The relation $a = B_N(b)$ is the link~\eqref{eq:nrlink},
$\alpha^{N-1} = \beta^{N-1}(1 + \tfrac{N-1}{N}b)$, so $\alpha^{N-1} = w_2$, which is the round-1
indifference $\Phi(a) = w_2$. Both conditions hold, so $(a, b)$ is an equilibrium. For the final
claim, every interior equilibrium satisfies $a = B_N(b)$ by Lemma~\ref{lem:nrlink} and its corollary
and $\Psi_N(b) = 0$ by its round-2 condition, so the two sets coincide.
\end{proof}

For a candidate profile $(a, b)$, suppose the other players use it, so $H$ is their common
submission CDF, and define the round-2 gap
\[
D(a, b) \;=\; \int_0^1 H^{n}\dd u - c - \beta^{n},
\]
where $n$ is the number of rivals to beat. The first part is the value of continuing to the last
draw against these opponents; the second, $\beta^n = H(b)^n$, is the win probability of accepting the
value $b$ against them. So $D$ is how much a player gains by redrawing rather than accepting $b$ when
the other players play $(a, b)$: positive means $b$ is too low, negative too high, and $D = 0$ makes
$b$ a best response to them. The next two lemmas sign its partial derivatives $\partial D/\partial a$
and $\partial D/\partial b$.

\begin{lemma}\label{lem:nrFb}
For every $n \geq 1$ and $0 < b \leq a < 1$, $\partial D/\partial b < 0$.
\end{lemma}

\begin{proof}
The partial of $H$ in $b$ is $au$ on $[0, b]$ and $a(u - 1) \leq 0$ on $[b, 1]$. Dropping the
negative part,
\[
\begin{aligned}
\frac{1}{n}\frac{\partial D}{\partial b}
\;&=\; \int_0^1 H^{n-1} H_b\dd u - 2ab\,\beta^{n-1}
\;\leq\; \int_0^b (abu)^{n-1}au\dd u - 2ab\,\beta^{n-1}\\
\;&=\; ab\,\beta^{n-1}\Bigl(\frac{b}{n+1} - 2\Bigr) < 0 .
\end{aligned}
\qedhere
\]
\end{proof}

\begin{lemma}\label{lem:nrFa}
For every $n \geq 1$ and $0 < b \leq a < 1$, $\partial D/\partial a < 0$ whenever
\begin{equation}\label{eq:nrsuff}
\bigl((1+b)a - b\bigr)^{n+1} \;<\; \bigl[(n+1)(1+b) - b\bigr]\, b^{2n} .
\end{equation}
In particular: (i) for $n = N - 1$, \eqref{eq:nrsuff} holds whenever $a = B_N(b)$;
(ii) for $n = N$, \eqref{eq:nrsuff} holds whenever $a \leq B_N(b)$.
\end{lemma}

\begin{proof}
Differentiate $D$ in $a$. The boundary term is $-\tfrac{1}{n}\partial_a\beta^n = -\beta^{n-1}b^2$,
since $\beta = ab^2$, so
\[
\frac{1}{n}\frac{\partial D}{\partial a}
\;=\; \int_0^1 H^{n-1}\,H_a\dd u \;-\; \beta^{n-1}b^2 .
\]
The kernel $H_a$ has two forms. Below $a$, both pieces of $H$ in~\eqref{eq:nrH3} are proportional to
$a$, so $H_a = H/a$ there. Above $a$, $H = (1+a+ab)u - a(1+b)$, so $H_a = (1+b)(u-1) \leq 0$.
Dropping the nonpositive upper part,
\[
\frac{1}{n}\frac{\partial D}{\partial a}
\;\leq\; \frac{1}{a}\int_0^a H^{n}\dd u \;-\; \beta^{n-1}b^2 .
\]
The segment formula used for~\eqref{eq:nrPsiA} (on a piece of slope $s$ from height $h_0$ to $h_1$,
$\int H^n \dd u = (h_1^{n+1} - h_0^{n+1})/((n+1)s)$) evaluates the integral over the two pieces
below $a$, with slopes $ab$ and $a(1+b)$ and heights $0$, $\beta$, $\alpha$:
\[
\frac{1}{a}\int_0^a H^{n}\dd u
\;=\; \frac{1}{n+1}\Bigl[\frac{\beta^{n+1}}{a^2 b} + \frac{\alpha^{n+1} - \beta^{n+1}}{a^2(1+b)}\Bigr]
\;=\; \frac{\alpha^{n+1}}{a^2(1+b)(n+1)} + \frac{\beta^{n-1}b^3}{(1+b)(n+1)},
\]
the second equality combining the $\beta^{n+1}$ terms
($\tfrac1b - \tfrac{1}{1+b} = \tfrac{1}{b(1+b)}$) and using $\beta^2/(a^2b) = b^3$. So
\[
\frac{1}{n}\frac{\partial D}{\partial a}
\;\leq\; \frac{\alpha^{n+1}}{a^2(1+b)(n+1)}
\;-\; \beta^{n-1}b^2\Bigl[1 - \frac{b}{(1+b)(n+1)}\Bigr].
\]
The right side is negative exactly when $\alpha^{n+1} < a^2\beta^{n-1}b^2\bigl[(n+1)(1+b) -
b\bigr]$, and substituting $\alpha = a\bigl((1+b)a - b\bigr)$ and $a^2\beta^{n-1}b^2 =
a^{n+1}b^{2n}$, the factor $a^{n+1}$ cancels and the condition is~\eqref{eq:nrsuff}.

(i) On $a = B_N(b)$, $(1+b)a - b = b^2\varphi_N$, so at $n = N-1$ the left side
of~\eqref{eq:nrsuff} is $b^{2N}\varphi_N^N$ and, after cancelling $b^{2N-2}$, the condition reads
$b^2\varphi_N^N < N + (N-1)b$. By definition $\varphi_N^{N-1} = (N + (N-1)b)/N$, so
$\varphi_N^N = \varphi_N\,(N + (N-1)b)/N$, and the factor $N + (N-1)b$ cancels from both sides:
the condition is exactly
\[
b^2\varphi_N \;<\; N .
\]
This holds with room: $\varphi_N^{N-1} = 1 + \tfrac{N-1}{N}b < 2$ gives $\varphi_N \leq 2$, and
$b^2 < 1$, so $b^2\varphi_N < 2 \leq N$.

(ii) The left side of~\eqref{eq:nrsuff} increases in $a$, so it suffices to check it at
$a = B_N(b)$, where it equals $b^{2(n+1)}\varphi_N^{\,n+1}$. At $n = N$, after cancelling
$b^{2N}$ and simplifying $(N+1)(1+b) - b = N(1+b) + 1$, the condition reads
\[
b^2\varphi_N^{\,N+1} \;<\; N(1+b) + 1 .
\]
For $N \geq 3$, raise $\varphi_N^{N-1} = 1 + \tfrac{N-1}{N}b < 1 + b$ to the power
$(N+1)/(N-1)$:
\[
\varphi_N^{\,N+1} \;<\; (1+b)^{(N+1)/(N-1)} \;\leq\; (1+b)^2 ,
\]
the second inequality because the exponent $(N+1)/(N-1)$ is at most $2$ exactly when $N \geq 3$
and the base exceeds one. Multiplying by $b^2 < 1$ and then comparing with the target,
\[
b^2\varphi_N^{\,N+1} \;<\; (1+b)^2 \;\leq\; 2(1+b) \;<\; N(1+b) + 1,
\]
using $1 + b \leq 2$ for the middle step and $N(1+b) + 1 - 2(1+b) = (N-2)(1+b) + 1 > 0$ for the
last. For $N = 2$:
$\varphi_2 = 1 + b/2 \leq 3/2$, so $b^2\varphi_2^{\,3} \leq 27b^2/8$, and $27b^2/8 < 3 + 2b$ on
$(0, 1]$ because the parabola $27b^2 - 16b - 24$ is convex and negative at both $b = 0$ and
$b = 1$.
\end{proof}

\subsection*{Uniqueness}

\begin{lemma}\label{lem:nrendpoint}
For every $N \geq 2$ and $c \in [0, 1/N)$,
\[
\Psi_N(0^+) \;=\; \frac{1}{N} - c \;>\; 0,
\qquad
\Psi_N(\bar b_N) \;=\; -\,\frac{N-1}{N\,(1 + \bar b_N)} \;-\; c \;<\; 0 .
\]
\end{lemma}

\begin{proof}
As $b \to 0$ the profile accepts every draw, $H \to u$, $\beta \to 0$, and
$\Psi_N \to 1/N - c$. At $\bar b_N$, $B_N(\bar b_N) = 1$, so $a = 1$, $\alpha = 1$, and the top segment
vanishes; then $\int_0^1 H^{N-1}\dd u = (1 + \bar b^{\,2N-1})/(N(1+\bar b))$. Raising the boundary
relation $\bar b^2 \varphi_N(\bar b) = 1$ to the power $N - 1$ gives
$\bar b^{\,2(N-1)} = N/(N + (N-1)\bar b)$, and substituting,
\[
\Psi_N(\bar b_N) + c
\;=\; \frac{1 + \bar b^{\,2N-1}}{N(1+\bar b)} - \bar b^{\,2(N-1)}
\;=\; \frac{1 - \bar b^{\,2(N-1)}\bigl(N + (N-1)\bar b\bigr)}{N(1+\bar b)}
\;=\; -\,\frac{N-1}{N(1+\bar b)} . \qedhere
\]
\end{proof}

\begin{proposition}\label{prop:nrunique}
For every $N \geq 2$ and $c \in [0, 1/N)$, $\Psi_N$ is strictly decreasing on $(0, \bar b_N]$.
Consequently the three-draw contest has exactly one symmetric equilibrium, namely
$(B_N(b^*), b^*)$ with $b^*$ the unique root of $\Psi_N$, and the equilibrium is a smooth function
of $(N, c)$, treating $N$ as a continuous variable in the scalar equation.
\end{proposition}

\begin{proof}
By $\Psi_N(b) = D(B_N(b), b)$ at $n = N - 1$, the chain rule gives
$\Psi_N' = D_a\,B_N' + D_b$. Then $D_b < 0$ by Lemma~\ref{lem:nrFb};
$D_a < 0$ on $a = B_N(b)$ by Lemma~\ref{lem:nrFa}(i); and $B_N' > 0$, since $b/(1+b)$ and
$b^2\varphi_N(b)/(1+b)$ are both strictly increasing. Hence $\Psi_N' < 0$. With
Lemma~\ref{lem:nrendpoint}, $\Psi_N$ has a unique root, which is the unique interior symmetric
equilibrium by the reduction; boundary equilibria are excluded by Lemma~\ref{lem:nrinterior}. Smoothness in $(N, c)$ follows from the implicit function theorem, since
$\Psi_N' < 0$ at the root.
\end{proof}

\subsection*{Selectivity of both thresholds}

Throughout this section $c = 0$; the extension to low positive cost is Remark~\ref{rem:nrsmallc}.
Write $(a_N, b_N)$ for the $N$-player equilibrium, with $a_N = B_N(b_N)$, and $\alpha_N = H(a_N)$,
$\beta_N = H(b_N)$ for its two heights; let $Z = H(U)$ be the submitted rank of a uniform draw
against it. Adding a player compares this equilibrium (the old one) with the $(N+1)$-player
equilibrium $(a_{N+1}, b_{N+1})$ (the new one).

\begin{proposition}\label{prop:nrselect}
For every $N \geq 2$, $\Psi_{N+1}(b_N) > 0$. Consequently $b_{N+1} > b_N$: the round-2 threshold of
the unique equilibrium strictly rises when a player joins.
\end{proposition}

\begin{proof}
\emph{Step 1: $B_N$ falls in $N$.} Write $\varphi_N(b) = e^{g(t)}$ with $t = N-1$,
$g(t) = \ln(1 + x(t))/t$, and $x(t) = bt/(t+1)$. Then
\[
t^2 g'(t) \;=\; \frac{x}{(t+1)(1+x)} - \ln(1+x) \;<\; \frac{x}{1+x} - \ln(1+x) \;\leq\; 0,
\]
so $\varphi_N$ is strictly decreasing in $N$ and $b \leq B_{N+1}(b) < B_N(b)$.

\emph{Step 2: the $(N+1)$-player gap falls in $a$ below the old map, $a \leq B_N(b)$.} This is
Lemma~\ref{lem:nrFa}(ii) with $n = N$.

\emph{Step 3: sign at $b_N$.} By Steps 1 and 2, moving from $B_N(b_N)$ down to $B_{N+1}(b_N)$ at
fixed $b_N$ raises the $(N+1)$-player gap,
\[
\Psi_{N+1}(b_N) \;=\; D\bigl(B_{N+1}(b_N), b_N\bigr) \;>\; D\bigl(B_N(b_N), b_N\bigr).
\]
At the old
equilibrium the round-2 condition reads $\mathbb{E}[Z^{N-1}] = \beta_N^{N-1}$, and $Z$ is
nondegenerate because $H$ is strictly increasing, so Lyapunov's inequality gives
\[
\mathbb{E}[Z^N] \;>\; \bigl(\mathbb{E}[Z^{N-1}]\bigr)^{N/(N-1)} \;=\; \beta_N^N,
\]
that is, $D(B_N(b_N), b_N) > 0$.

By Lemma~\ref{lem:nrendpoint} applied at contest size $N+1$, $\Psi_{N+1}(\bar b_{N+1}) < 0$, and by
Proposition~\ref{prop:nrunique} the function $\Psi_{N+1}$ is strictly decreasing with unique root
$b_{N+1}$; positivity at $b_N$ therefore places $b_{N+1}$ strictly above $b_N$.
\end{proof}

\begin{proposition}\label{prop:nrtransfer}
For every $N \geq 2$, $a_{N+1} > a_N$: the round-1 threshold of the unique equilibrium also
strictly rises when a player joins.
\end{proposition}

\begin{proof}
The round-1 threshold satisfies $a_{N+1} = B_{N+1}(b_{N+1})$, and $B_{N+1}$ is increasing, so
$a_{N+1} > a_N$ is equivalent to $b_{N+1} > \tilde b$, where $\tilde b$ solves
$B_{N+1}(\tilde b) = a_N$: the round-2 threshold that reproduces the old round-1 threshold in the new
contest. Since $\Psi_{N+1}$ is strictly decreasing with root $b_{N+1}$, it suffices to show
$\Psi_{N+1}(\tilde b) > 0$, which places $b_{N+1}$ above $\tilde b$.

This $\tilde b$ is well defined, with $b_N < \tilde b \leq a_N$. The map $B_{N+1}$ is continuous and
strictly increasing, hence invertible, and Step 1 of Proposition~\ref{prop:nrselect} locates $a_N$ in
its range: at $b_N$, $B_{N+1}(b_N) < B_N(b_N) = a_N$, so $B_{N+1}$ has not yet reached $a_N$ and
$\tilde b > b_N$; at $a_N$, $B_{N+1}(a_N) \geq a_N$, so $B_{N+1}$ has already reached $a_N$ and
$\tilde b \leq a_N$.

Define the round-1 gap of the $(N+1)$-player contest at a profile $(a, b)$,
\[
J(a, b) \;=\; \int_0^1 \max\{H(u), \beta\}^N\dd u - \alpha^N
\;=\; b\,\beta^N + \int_b^1 H^N\dd u - \alpha^N .
\]

\emph{Step 1: $J$ agrees with the round-2 gap on $a = B_{N+1}(b)$.} Using $H = s_1 u$ on
$[0, b]$,
\[
J - D \;=\; \int_0^b(\beta^N - H^N)\dd u - (\alpha^N - \beta^N)
\;=\; \frac{N b\,\beta^N}{N+1} - (\alpha^N - \beta^N),
\]
and on $a = B_{N+1}(b)$, $\alpha^N = \beta^N\varphi_{N+1}^N = \beta^N(1 + \tfrac{N}{N+1}b)$, so
$J = D$ there; in particular $\Psi_{N+1}(\tilde b) = D(a_N, \tilde b) = J(a_N, \tilde b)$.

\emph{Step 2: $J > 0$ at the old equilibrium.} The old round-1 condition reads
$\alpha_N^{N-1} = \mathbb{E}[X^{N-1}]$, where $X = \max\{Z, \beta_N\}$ is nondegenerate.
Lyapunov's inequality then gives $\mathbb{E}[X^N] > \alpha_N^N$, that is, $J(a_N, b_N) > 0$.

\emph{Step 3: $J(a_N, \cdot)$ increases on $[b_N, \tilde b]$.} Differentiate $J$ in its second
argument $b$ at fixed first argument $a_N$, from the form $J = b\beta^N + \int_b^1 H^N\dd u -
\alpha^N$. The three summands differentiate to
\[
\begin{aligned}
\frac{\partial}{\partial b}\bigl(b\beta^N\bigr) &= \beta^N + 2Nab^2\beta^{N-1},\\
\frac{\partial}{\partial b}\int_b^1 H^N\dd u &= -\beta^N - Na\int_b^1 H^{N-1}(1-u)\dd u,\\
\frac{\partial}{\partial b}\bigl(-\alpha^N\bigr) &= Na(1-a)\alpha^{N-1},
\end{aligned}
\]
using $\partial\beta/\partial b = 2ab$, $\partial\alpha/\partial b = -a(1-a)$, and $H_b = a(u-1)$ on
$[b, 1]$; the integral's lower-limit term is $-H(b)^N = -\beta^N$. Adding the three, the two
$\beta^N$ cancel, leaving
\[
\frac{\partial J}{\partial b}
\;=\; N a\Bigl[\,2b^2\beta^{N-1} - \int_b^1 H^{N-1}(1-u)\dd u + \alpha^{N-1}(1-a)\Bigr].
\]
The first and third bracket terms are positive; we show they outweigh the subtracted integral.
Evaluate at $b = s \in [b_N, \tilde b]$, profile $(a_N, s)$, writing $Z$ for the submitted rank $H(U)$ at this profile, and bound the integral in two moves.
Since $1 - u \leq 1 - s$ on $[s, 1]$,
\[
\int_s^1 H^{N-1}(1-u)\dd u \;\leq\; (1-s)\int_0^1 H^{N-1}\dd u \;=\; (1-s)\,\mathbb{E}[Z^{N-1}];
\]
and $\mathbb{E}[Z^{N-1}] \leq \beta^{N-1}$, because $\mathbb{E}[Z^{N-1}] - \beta^{N-1}$ is the
$N$-player round-2 gap $D(a_N, s)$, which is $0$ at $s = b_N$ and falls in $b$ by
Lemma~\ref{lem:nrFb}, hence $\leq 0$ for $s \geq b_N$. With $\alpha \geq \beta$, factoring out
$\beta^{N-1}$,
\[
\frac{\partial J}{\partial b}
\;\geq\; N a_N\,\beta^{N-1}\bigl[\,2s^2 - (1 - s) + (1 - a_N)\bigr]
\;=\; N a_N\,\beta^{N-1}\bigl[\,2s^2 - (a_N - s)\bigr].
\]
This is positive because $a_N - s$ is small. With $a_N = B_N(b_N)$, the closed
form~\eqref{eq:nrbranch} gives $B_N(b_N) - b_N = b_N^2(\varphi_N(b_N) - 1)/(1 + b_N)$, so
\[
a_N - s \;\leq\; a_N - b_N \;=\; \frac{b_N^2(\varphi_N(b_N) - 1)}{1 + b_N}
\;\leq\; \frac{b_N^2}{2} \;\leq\; \frac{s^2}{2}
\]
(the first and last inequalities from $s \geq b_N$, the middle from $\varphi_N \leq 3/2$ and
$1 + b_N \geq 1$). Hence $2s^2 - (a_N - s) \geq \tfrac{3}{2}s^2$, and
$\partial J/\partial b \geq \tfrac{3}{2}N a_N s^2\beta^{N-1} > 0$.

Combining Steps 2 and 3, $\Psi_{N+1}(\tilde b) = J(a_N, \tilde b) \geq J(a_N, b_N) > 0$.
\end{proof}

\begin{proof}[Proof of Theorem~\ref{thm:nr_k3}]
Uniqueness and the closed form are Proposition~\ref{prop:nrunique}; the comparative statics at
$c = 0$ are Propositions~\ref{prop:nrselect} and~\ref{prop:nrtransfer}; the extension in $c$ is
Remark~\ref{rem:nrsmallc}.
\end{proof}

\begin{remark}\label{rem:nrsmallc}
By Proposition~\ref{prop:nrunique}, the equilibrium $b^*(N, c)$ is smooth in $c$ on $[0, 1/N)$, so
the evaluation points $b_N(c)$, $a_N(c)$, and $\tilde b(c)$ move continuously with $c$. Every
inequality in the proofs of Propositions~\ref{prop:nrselect} and~\ref{prop:nrtransfer} is strict at
$c = 0$ and continuous in $c$ along these paths, and the cost cancels from Lemma~\ref{lem:nrlink},
so each conclusion persists on an interval $[0, \bar c_N)$, with $\bar c_N < 1/(N+1)$ so that the $(N+1)$-player contest is admissible as well.
\end{remark}

\medskip\noindent
Proposition~\ref{prop:nr_k3red} collects Lemma~\ref{lem:nrlink} (part i), Corollary~\ref{cor:nrbranch} (part ii), and Lemma~\ref{lem:nrconverse} together with the display~\eqref{eq:nrPsiA} (part iii); the function $\Psi_N$ there coincides with~\eqref{eq:nrPsi}.

\section{Derivation of the Designer FOC}\label{app:designer}

Consider $M$ symmetric designers, each with $N$ workers using threshold~$b$. Each team's maximum has CDF $H(x) = G(x)^N$ where $G(x) = (F(x)-F(b))/(1-F(b))$ for $x\geq b$.

Designer $A$ deviates to threshold $b_A$. Her team max has CDF $H_A(x)=G_A(x)^N$ where $G_A(x)=(F(x)-F(b_A))/(1-F(b_A))$. Her winning probability is
\[
P(b_A) = \int_b^K H(x)^{M-1}\,N\,G_A(x)^{N-1}\frac{f(x)}{1-F(b_A)}\,dx.
\]
We differentiate with respect to $b_A$ and evaluate at $b_A=b$. The derivative of $G_A(x)^{N-1}\cdot f(x)/(1-F(b_A))$ with respect to $b_A$ is
\[
\frac{f(b)f(x)}{(1-F(b))^2}\,u^{N-2}\bigl[Nu-(N-1)\bigr],
\]
where $u=G(x)=(F(x)-F(b))/(1-F(b))$. Since $G(x)^{N(M-1)} = u^{N(M-1)}$ and $f(x)\,dx = (1-F(b))\,du$, we obtain
\begin{align*}
\frac{\partial P}{\partial b_A}\bigg|_{b_A=b}
&= \frac{Nf(b)}{1-F(b)}\int_0^1 u^{NM-2}\bigl[Nu-(N-1)\bigr]\,du \\
&= \frac{Nf(b)}{1-F(b)}\left[\frac{N}{NM}-\frac{N-1}{NM-1}\right] \\
&= \frac{Nf(b)}{1-F(b)}\cdot\frac{M-1}{M(NM-1)}.
\end{align*}
Setting $\Omega\cdot\partial P/\partial b_A = Ncf(b)/(1-F(b))^2$ yields $1-F(b) = cM(NM-1)/(\Omega(M-1))$. The payoff and dissipation formulas follow: each designer wins with probability $1/M$ and pays her $N$ workers' expected search costs $Nc/(1-F(b_D)) = N\Omega(M-1)/(M(NM-1))$, so her payoff is $\Omega/M - N\Omega(M-1)/(M(NM-1)) = \Omega(N-1)/(M(NM-1))$; total expenditure is $M$ times the cost term, $N\Omega(M-1)/(NM-1)$, and the dissipation ratio is $N(M-1)/(NM-1)$.